\documentclass[11pt,a4paper]{article}
\usepackage[utf8]{inputenc}

\newsavebox\ideabox
\newenvironment{idea}
  {\begin{equation}
   \begin{lrbox}{\ideabox}
   \begin{minipage}{\dimexpr\columnwidth-2\leftmargini}
   \setlength{\leftmargini}{0pt}%
   \begin{quote}}
  {\end{quote}
   \end{minipage}
   \end{lrbox}\makebox[0pt]{\usebox{\ideabox}}
   \end{equation}}

\date{\vspace{-5ex}}
\usepackage{styling}

\hypersetup{
  pdftitle={The Power of Filling in Balanced Allocations},
  pdfauthor={Dimitrios Los, Thomas Sauerwald, John Sylvester}
  }

\title{\textbf{The Power of Filling in Balanced Allocations}\footnote{Some results from the paper were presented at SODA 2022 \cite{OurSODA}.}}

\author[1]{Dimitrios Los}  
\author[1]{Thomas Sauerwald}
\author[2]{John Sylvester}
\affil[1]{Department of Computer Science \& Technology, University of Cambridge, UK\\ \texttt{firstname.lastname@cl.cam.ac.uk}} 
\affil[2]{Department Of Computer Science, University of Liverpool, UK\\
	\texttt{john.sylvester@liverpool.ac.uk}}

\newcommand{\WOne}{{\hyperlink{w1}{\ensuremath{\mathcal{W}}}}\xspace}

\newcommand{\POne}{{\hyperlink{p1}{\ensuremath{\mathcal{P}}}}\xspace}

\begin{document}
 
\maketitle

\begin{abstract}

We introduce a new class of balanced allocation processes which are primarily characterized by ``filling'' underloaded bins. A prototypical example is the \textsc{Packing} process: At each round we only take one bin sample, if the load is below the average load, then we place as many balls until the average load is reached; otherwise, we place only one ball. We prove that for any process in this class the gap between the maximum and average load is $\mathcal{O}(\log n)$ w.h.p.\ for any number of balls $m\geq 1$. For the \textsc{Packing} process, we also provide a matching lower bound. Additionally, we prove that the \textsc{Packing} process is sample-efficient in the sense that the expected number of balls allocated per sample is strictly greater than one. Finally, we also demonstrate that the upper bound of $\mathcal{O}(\log n)$ on the gap can be extended to the \textsc{Memory} process studied by Mitzenmacher, Prabhakar and Shah (2002).

	\medskip


\end{abstract}

\section{Introduction}

We consider the sequential allocation of $m$ balls (jobs or data items) to $n$ bins (servers or memory cells), by allowing each ball to choose from a set of randomly sampled bins. The goal is to allocate balls efficiently, while also keeping the load distribution balanced. The balls-into-bins framework has found numerous applications in 
hashing, load balancing, routing, but is also closely related to more theoretical topics such as randomized rounding or pseudorandomness (we refer to the surveys~\cite{MR01} and~\cite{W17} for more details).

A classical algorithm is the \DChoice process introduced by Azar, Broder, Karlin \& Upfal~\cite{ABKU99} and Karp, Luby \& Meyer auf der Heide~\cite{KLM96}, where for each ball to be allocated, we sample $d \geq 1$ bins uniformly and then place the ball in the least loaded of the $d$ sampled bins. It is well-known that for the \OneChoice process ($d=1$), the gap between the maximum and average load is $\Theta\bigr( \sqrt{ m/n \cdot \log n } \bigr)$ \Whp, when $m \geq n \log n$. In particular, this gap grows significantly as $m \gg n$. For $d=2$,~\cite{ABKU99} proved that the gap is only $\log_2 \log n+\Oh(1)$ \Whp\ for $m=n$. This result was generalized by Berenbrink,  Czumaj, Steger \& V\"{o}cking~\cite{BCSV06} who proved that the same guarantee also holds for $m \geq n$, which is called the {\em heavily loaded case}. This means that the gap between the maximum and average load remains a slowly growing function in $n$ that is independent of $m$.
This dramatic improvement of \TwoChoice over \OneChoice is widely known as the ``power of two choices''.

It is natural to investigate allocation processes that are less powerful than \TwoChoice in their ability to sample two bins, sample uniformly or distinguish between the load of the sampled bins. Such processes make fewer assumptions than \TwoChoice and can thus be regarded as more sample-efficient and robust. As noted in~\cite{LXKGLG11}, communication is a shortcoming of \TwoChoice in some real-word systems:

\begin{idea}\label{eq:22quote}
	\textit{More importantly, the \textsc{[Two-Choice]} algorithm requires communication between dispatchers and processors at the time of job assignment. The communication time is on the critical path, hence contributes to the increase in response time.}
\end{idea}

One example of such an allocation process is the $(1+\beta)$-process analyzed by Peres, Talwar \& Wieder~\cite{PTW15}, where each ball is allocated using \OneChoice with probability $1-\beta$ and otherwise is allocated using \TwoChoice. The authors proved that for any $\beta \in (0,1]$, the gap is $\Oh\big(\frac{\log n}{\beta}+\frac{\log(1/\beta)}{\beta}\big)$ \Whp; thus for any $\beta=\Omega(\poly(1/n))$, the gap is $\Oh\big(\frac{\log n}{\beta}\big)$. Hence, only a ``small'' fraction of \TwoChoice rounds are enough to inherit the property of \TwoChoice that the gap is independent of $m$. A similar class of adaptive sampling schemes (where, depending on the loads of the samples so far, the ball may or may not sample another bin) was analyzed by Czumaj and Stemann~\cite{CS01}, but their results hold only for $m=n$.

Another important process is \TwoThinning, which was studied in~\cite{IK04,FG18}. In this process, each ball first samples a bin uniformly and based on some criterion (e.g., a threshold on the load) decides whether to allocate the ball there or not. If not, the ball is placed into a {\em second} bin sample, without inspecting its load. In~\cite{FG18}, the authors proved that for $m=n$, there is a \TwoThinning process which achieves a gap of $\Oh\big({ \scriptstyle \sqrt{ \frac{\log n}{\log \log n}}} \big)$ \Whp. A protocol is introduced and analyzed in \cite{AugustineMRU16} which utilizes $d$ adaptive bin samples (on average) and improves significantly on the gap bound of \DChoice. These results present a vast improvement over \OneChoice, however the total number of samples is greater than one per ball. Similar threshold processes have been studied in queuing~\cite{ELZ86}, \cite[Section 5]{M96} and discrepancy theory~\cite{DFG19}.
For values of $m$ sufficiently larger than $n$, \cite{FGL21} and \cite{LS21} prove some lower and upper bounds for a more general class of \emph{adaptive} thinning protocols (here, adaptive means that the choice of the threshold \emph{may} depend on the load configuration). Related to this line of research, the authors of \cite{LS21} also analyze a so-called \Quantile-process, which is a version of \TwoThinning where the ball is placed into a second sample only if the first bin has a load which is at least the median load.

Finally, we mention the \Memory process analyzed by Mitzenmacher, Prabhakar \& Shah~\cite{MPS02}, which is essentially a version of the \TwoChoice process with a cache. At each round, we take a uniform bin sample but we also have access to a cache. Then the ball is placed in the least loaded of the sampled and the cached bin, and after that the cache is updated if needed. It was shown in~\cite{MPS02} that for $m=n$, the process achieves an asymptotically better gap than \TwoChoice. In this work we prove (to the best of our knowledge) the first $\Oh(\log n)$ bound on the gap in the heavily loaded case. Later, in \cite{SODA23}, the authors of this paper improved this to the asymptotically tight $\Oh(\log \log n)$ bound. Luczak \& Norris~\cite{LN13} analyzed the related ``supermarket'' model with memory, in the queuing setting. 

From a more technical perspective, apart from analyzing a large class of natural allocation processes, an important question is to understand how sensitive the gap is to changes in the allocation probability vector. To this end,~\cite{PTW15} formulated general conditions on the probability vector, which, when satisfied in all rounds, imply a small gap bound. These were then later applied not only to the $(1+\beta$)-process, but also to analyze a ``graphical'' allocation model where a pair of bins is sampled by picking an edge of a graph uniformly at random. 

Wieder~\cite{W07} studied the \DChoice process with an $(a, b)$-biased sampling vector $s$, meaning that the probability of sampling bin $i$ is $\frac{1}{an} \leq s_i \leq \frac{b}{n}$. They showed that for the \DChoice process there exist constants $A := A(d), B := B(d)$, so that for every $(a, b)$-biased sampling vector with $a < A$ and $b < B$, the process maintains the $\Oh(\log \log n)$ gap bound, otherwise, there are sampling vectors, for which the process has an $\Omega(\sqrt{\frac{m}{n} \cdot \log n})$ gap for $m \gg n$.

Several variants of balls-into-bins have been studied, including allocations on graphs~\cite{KP06,BF22} and hypergraphs~\cite{G08,GMP20}, balls-into-bins with correlated choices~\cite{W07}, balls-into-bins with hash functions~\cite{CRSW11} and balls-into-bins with deletions~\cite{BK22,CFMMRSU98}.

\paragraph{Brief Summary of Our Results} 
In this work, we address the shortcoming of \TwoChoice observed in~\cite{LXKGLG11} (see the quote \eqref{eq:22quote} above), by studying processes that can allocate more than one balls to a sampled bin. While most of the allocation processes studied before are based on a ``sampling bias'' towards underloaded bins, our new class consists of  processes with very weak sampling requirements (e.g., uniform sampling like in~\OneChoice suffices). However, our class is based on a so-called
``filling operation'', meaning in the most basic form means that once an undeloaded bin is found, it is filled with balls until the average load is reached. This avoids the use of a second sample (like in \TwoChoice, $(1+\beta)$-process or \Thinning) and also does not require the allocator to hold and compare the load of two sampled bins (like in \TwoChoice or $(1+\beta)$-process). Therefore, due to the reduced sampling and communication costs, this operation seems well suited to scenarios where an allocator needs to quickly assign a large number of jobs to a set of servers.

In order to capture not only the prototypcial \Packing process but also a variety of other processes, we formulate two conditions, \POne and \WOne, which give rise to a broader class of allocation processes with this ``filling'' behavior, we call these \Filling processes. Very roughly, \POne  stipulates that a procedure is used to sample a single bin in each round that is not more biased towards overloaded bins than the uniform distribution. Again, very roughly, \WOne states that if we sample an underloaded bin then we allocate the amount of balls that would bring that bin up to the average load, however we may distribute them \textit{almost} arbitrarily over the underloaded bins; otherwise, we allocate one ball. 

As our first main result (\Cref{thm:filling_key}), we prove that if \POne \emph{and} \WOne both hold, then \Whp\ for any round $m$, a gap bound of $\Oh(\log n)$ follows.  Note that from here on, we will generally use $m$ to denote the number of rounds, which may be different (i.e., smaller than) the number of total balls allocated. While it is easy to show that \Packing meets the two conditions, some care is needed to apply the framework to \Memory due to its use of the cache, which creates strong correlations between the allocation of any two consecutive balls. The ability to allocate the balls arbitrarily among underloaded bins, means that the framework directly captures versions of the \Packing process augmented with simple heuristics. For example, if the bins are uniformly accessed but also locally networked, then a bin could spread the balls it receives in a round to neighboring underloaded bins.

We also show that any process satisfying  \POne and \WOne allocates $1+c$ balls (for some constant $c > 0$) per round in expectation (\Cref{thm:sample_efficient}). A direct consequence of this is that \Packing is more sample efficient than \OneChoice, while still achieving an $\Oh(\log n)$ gap. This matches the gap bound of the $(1+\beta)$-process for constant $\beta \in (0, 1)$, which requires strictly more than one sample per ball, demonstrating the ``power of filling'' in balanced allocations.  We further investigate this phenomenon by analyzing two variants of \Packing: $(i)$ \TightPacking, where the filling balls are adversarially allocated to the highest underloaded bins  and $(ii)$ \Packing with an $(a, b)$-biased sampling vector, where the bins may be selected using a sampling vector that majorizes \OneChoice. We show that in contrast to \DChoice, for arbitrary $a, b > 1$ being functions of $n$, the gap of this process is independent of $m$.

For \Packing, we also prove a matching lower bound on the gap of $\Omega(\log n)$ for any $m = \Omega(n \log n)$. Our results on the gap are summarized in \Cref{tab:overview}. We note that this work first appeared as part of the conference paper \cite{OurSODA}, the results on \textsc{Non-Filling} processes from \cite{OurSODA} have also subsequently been refined \cite{los2023meanbiased}.

\begin{table}[h]	\resizebox{\textwidth}{!}{
		\renewcommand{\arraystretch}{1.5}
		\centering
		\begin{tabular}{|c|cc|cc|cc|cc|}
			\hline 
			\multirow{2}{*}{Process} & \multicolumn{4}{c|}{Lightly Loaded Case $m=\Oh(n)$} & \multicolumn{4}{c|}{Heavily Loaded Case $m \geq n$} \\ \cline{2-9}
			& \multicolumn{2}{c|}{Lower Bound} & \multicolumn{2}{c|}{Upper Bound} & \multicolumn{2}{c|}{Lower Bound} & \multicolumn{2}{c|}{Upper Bound}
			\\ \hline 
			\makecell{$(1+\beta)$-process \\ with $\beta = 1/2$} & \multicolumn{4}{c|}{\phantom{[14]}$\qquad \quad\;\,\frac{\log n}{\log \log n}\qquad\quad\;\,$\cite{RS98} } & \multicolumn{4}{c|}{\phantom{[23]}$\qquad \quad\;\,\ \log n\qquad\quad\;\,$\cite{PTW15}} \\ 	
			\hhline{|-|-|-|-|-|-|-|-|-|}
			\multirow{2}{*}{\Memory} & \multicolumn{4}{c|}{\multirow{2}{*}{\phantom{[12]}$\qquad \quad\log \log n\quad \qquad $\cite{MPS02}}} & \multirow{2}{*}{$\log\log n$}  & \multirow{2}{*}{\cite{SODA23}} & \cellcolor{Greenish} $\log n$  & \cellcolor{Greenish} \footnotesize{(Thm~\ref{thm:filling_key})} \\ 
			\cline{8-9}
			& \multicolumn{4}{c|}{} & &  & $\log\log n$  & \cite{SODA23} \\ 
			\hhline{|-|-|-|-|-|-|-|-|-|}
			\Packing & \cellcolor{Greenish}$\frac{\log n}{\log \log n}$ &  \cellcolor{Greenish} \footnotesize{(Lem~\ref{lem:uniform_prob_lower_bound})} & \cellcolor{Greenish}$\frac{\log n}{\log \log n}$& \cellcolor{Greenish}\footnotesize{(Obs~\ref{obs:upper_bound_lightly_loaded})} & \cellcolor{Greenish}$\log n$ & \cellcolor{Greenish}\footnotesize{(Lem~\ref{lem:packing_logn_lb})} & \cellcolor{Greenish}$\log n$ & \cellcolor{Greenish}\footnotesize{(Thm~\ref{thm:filling_key})} \\
			\hhline{|-|-|-|-|-|-|-|-|-|}
			\TightPacking & \cellcolor{Greenish}$\frac{\log n}{\log \log n}$ & \cellcolor{Greenish} \footnotesize{(Lem~\ref{lem:uniform_prob_lower_bound})} & \cellcolor{Greenish}$\frac{\log n}{\log \log n}$ & \cellcolor{Greenish}\footnotesize{(Obs~\ref{obs:upper_bound_lightly_loaded})} &\cellcolor{Greenish}$ \frac{\log n}{\log \log n}$ & \cellcolor{Greenish}\footnotesize{(Lem~\ref{lem:uniform_prob_lower_bound})} & \cellcolor{Greenish}$\log n$ & \cellcolor{Greenish}\footnotesize{(Thm~\ref{thm:filling_key})} \\
			\hhline{|-|-|-|-|-|-|-|-|-|}
			\hline 
			
		\end{tabular}

	}
	\phantom{need to keep the table separate from caption }

	\caption{Overview of the Gap achieved (with probability at least $1-n^{-1}$ for upper bounds and at least a positive constant for lower bounds), by different allocation processes considered in this work (and others). All stated bounds hold asymptotically; upper bounds hold for all values of $m$, while lower bounds may only hold for certain values of $m$. Cells shaded in \hlgreenish{~Green~} are results proved in this paper. 
	}\label{tab:overview}
\end{table}

\paragraph{Proof Overview \& Ideas}
To prove an upper bound on the gap for \Filling processes, we consider an exponential potential function $\Phi^{t}$ with parameter $\alpha > 0$ (a variant of the one used in~\cite{PTW15, S77}). This potential function considers only bins that are overloaded by at least two balls; thus it is blind to the load configuration across underloaded bins. We upper bound $\Ex{ \Phi^{t+1} \, \mid \, \Phi^{t}}$ by an expression that is maximized if the process uses the uniform distribution for picking a bin (see \Cref{lem:filling}). We then use this upper bound to deduce that: $(i)$ in many rounds, the potential drops by a factor of at least $1-\Omega(\frac{\alpha}{n})$, and $(ii)$ in other rounds, the potential increases by a factor of at most $1+\mathcal{O}(\frac{\alpha^2}{n})$ (see \Cref{lem:filling_good_quantile}). Taking $\alpha$ sufficiently small and constructing a suitable 
super-martingale, we conclude that $\Ex{\Phi^{t}}=\Oh(n)$ for all rounds $t \geq 0$. The desired gap bound is then implied by Markov's inequality.
There are several challenges in extending the potential function approach of \cite{PTW15}. The conditions in \cite{PTW15} do not cover allocating a varying number of balls (dependent on the load distribution) in a single round, and some care is needed as in our processes the load of underloaded bins can change by more than a constant amount in each round. Additionally, for the processes considered in \cite{PTW15}, the potential drops in expectation in each round regardless of the load configuration. However, for our processes this is not the case (see Section \ref{sec:potentialgoesup} of the appendix for a concrete configuration which causes the potential to increase in expectation).

\paragraph{Organization}  In \Cref{sec:notation} we define some basic mathematical notation needed for our analysis and introduce some previously studied balanced allocation processes related to our work. In \Cref{sec:framework} we present our framework for \Filling processes and state our $\Oh(\log n)$ gap bound and sample efficiency bounds for \Packing. Following this we apply these general results to the \Packing, \TightPacking, and \Memory processes. In \Cref{sec:analysis_filling}, we prove an $\Oh(\log n)$ bound on the gap for \Filling processes and a bespoke bound for \Packing with an $(a, b)$-biased probability vector (which is not a \Filling process). In \Cref{sec:unfolding_general} we prove the gap bound for the so-called ``unfolding'' of a \Filling process, and then we prove that this applies to \Memory. The derivation of the lower bounds is given in \Cref{sec:lower}. Our bounds on the sample-efficiency (and a related quantity called throughput) of \Filling processes are then proved in \Cref{sec:sampleeff}. In \Cref{sec:experiemental_results} we empirically compare the gaps of the different processes. Finally, in \Cref{sec:conclusions} we conclude by summarizing our main results and pointing to some open problems.

\section{Notation and Preliminaries}\label{sec:notation}

Let $x^t$ denote load vector of the $n$ bins at round $t=0,1,2,\ldots$ (the state after $t$ rounds have been completed). So in particular, $x^0:=(0,\ldots,0)$. In many parts of the analysis, we will assume an arbitrary labeling of the $n$ bins so that their loads at round $t$ are ordered non-increasingly, i.e., 
\[
x_1^t \geq x_2^t \geq \cdots \geq x_n^t.
\]

We now give a formal definition of the classical \DChoice process~\cite{ABKU99,KLM96} for reference.  

\begin{framed}
	\vspace{-.45em} \noindent
	\underline{\DChoice Process:}\\
	\textsf{Iteration:} For each $t \geq 0$, sample $d$ bins $i_1,\ldots, i_d$ with replacement, independently and uniformly. Place a ball in a bin $i_{\min}$ satisfying $x_{i_{\min}}^t = \min_{1\leq j\leq d} x_{i_{j}}^t $, breaking ties arbitrarily.\vspace{-.5em}
\end{framed}
\noindent Mixing \OneChoice with \TwoChoice rounds at a rate $\beta$, one obtains the $(1+\beta)$-process~\cite{PTW15}:
\begin{framed}
	\vspace{-.45em} \noindent
	\underline{$(1+\beta)$-Process:}\\
	\textsf{Parameter:} A probability $\beta \in (0,1]$.\\
	\textsf{Iteration:} For each $t \geq 0$, with probability $\beta$ allocate one ball via the \TwoChoice process, otherwise allocate one ball via the \OneChoice process. \vspace{-.5em}
\end{framed}

Unlike the standard balls-into-bins processes, some of our processes may allocate more than one ball in a single round. To this end we define $W^t := \sum_{i\in[n]} x_i^t$ as the total number of balls that are allocated by round $t$ (for \DChoice we allocate one ball per round, so $W^t=t$; however our processes may allocate more than one ball in each round, that is $W^t \geq t$). Formally, an \textit{allocation process} satisfies that $(x_i^{t+1} - x_i^t) \in \N$ for any bin $i \in [n]$ (we only add balls) and $W^{t+1} \geq W^t + 1$ (we allocate at least one ball) at any round $t \geq 0$. For some processes, it makes sense to define $S^t$ as the number of bins sampled after $t$ rounds, thus for \Packing we have $S^t=t$ for all $t\geq 0$.

We define the gap as $\Gap(t):=\max_{i \in [n]} x_i^{t} - \frac{W^t}{n}$, which is the difference between the maximum load and average load at round $t$. When referring to the gap of a specific process $P$, we write $\Gap_{P}(t)$ but may simply write $\Gap(t)$ if the process under consideration is clear from the context.
Finally, we define the normalized load of a bin $i$ as:
\[
y_i^{t} :=  x_i^t - \frac{W^t}{n}.
\]
Further, let $B_{+}^{t}:=\left\{ i \in [n] \colon y_i^{t} \geq 0 \right\}$ be the set of overloaded bins, and $B_{-}^{t} := [n] \setminus B_{+}^{t}$ be the set of underloaded bins. 

Similarly to \cite{PTW15}, some of the allocation processes we study, sample a bin $i^t$ in each round $t$, using a probability vector $p^t$, where $p_j^{t}$ is the probability that the process samples the $j$-th most heavily loaded bin in round $t$. We denote this sampling a according $p^t$ by $i^t \in_{p^t} [n]$. A special case of a probability vector is the uniform vector of \OneChoice, which is $p_i^t = p_i = \frac{1}{n}$ for all $i  \in [n]$ and $t \geq 0$. A probability vector is called $(a, b)$-biased for $a, b \geq 1$ if $\frac{1}{an} \leq p_i \leq \frac{b}{n}$ for all $i \in [n]$. For two probability vectors $p$ and $q$ (or analogously, for two sorted load vectors), we say that $p$ majorizes $q$ if for all $1 \leq k \leq n$,
$
\sum_{i=1}^{k} p_i \geq \sum_{i=1}^k q_i.
$

We define $\mathfrak{F}^{t}$ as the filtration corresponding to the first $t$ allocations of the process (so in particular, $\mathfrak{F}^{t}$ reveals $x^{t}$ and $i^{t-1}$). In general, \textit{with high probability} (\Whp) refers to probability of at least $1 - n^{-c}$ for some constant $c > 0$. For random variables $Y,Z$ we say that $Y$ is stochastically smaller than $Z$ (or equivalently, $Y$ is stochastically dominated by $Z$)  if $ \Pro{ Y \geq x } \leq \Pro{ Z \geq x }$ for all real $x$. Throughout the paper, we often make use of statements and inequalities which hold only for sufficiently large $n$. For simplicity, we do not state this explicitly.

\section{Our Results on Filling Processes} \label{sec:framework}
In this section we shall present our main results. In the first subsection we present the framework for \Filling processes, then state the gap bound (\Cref{thm:filling_key}) and finally show they are sample-efficient (\Cref{thm:sample_efficient}) under the additional assumption of sampling each bin uniformly. In \Cref{sec:processes} we introduce the \Packing process, which is our prototypical example of a \Filling process, and give a lower bound showing our general gap bound is essentially best possible. The \TightPacking process is also introduced as an example of a more adversarial process that nevertheless still fits our framework, and thus also has a logarithmic gap. Finally in \Cref{sec:memory} we introduce the \Memory process and a method we call ``unfolding'' that can extend the applicability of our gap result. In particular, although \Memory is not a \Filling process we show that it is an unfolding of a suitable \Filling process, and thus we obtain a gap bound for \Memory for a polynomial number of balls. Finally in \Cref{sec:packing_with_negative_bias} we analyze the \Packing process with an $(a, b)$-biased probability vector, which is not a \Filling process but highlights the power of filling, as the gap is still independent of $m$.

\subsection{Filling Processes}
We now formally define the main class of processes studied in this paper. Recall that $B_{-}^{t}:=\left\{ i \in [n] \colon y_i^{t} < 0 \right\}$ is the set of underloaded bins after $t$ rounds have been completed. 

\begin{framed}
	\vspace{-.45em} \noindent
	\underline{\textsc{Filling}~Processes:}\\
	For each round $t \geq 0$,
	we sample a bin $i=i^t$ and then place a certain number of balls to $i$ (or other bins). More formally:
	\medskip 
	
	\noindent\textbf{Condition \hypertarget{p1}{$\mathcal{P}$}}: For each round $t\geq 0$, pick an arbitrary labeling of the $n$ bins such that $y_1^{t} \geq y_2^{t} \geq \cdots \geq y_n^{t}$. Then let $i = i^t \in_{p^t} [n]$, where $p^t$ is any probability vector which is majorized by the uniform distribution (\OneChoice). 
	
	\medskip

	\noindent\textbf{Condition \hypertarget{w1}{$\mathcal{W}$}}: For each round $t\geq 0$, with $i$ being the bin chosen above:
	\begin{itemize}
		\item If $y_i^{t} \geq 0$ then allocate a single ball to bin $i$.
		\item If $y_i^{t} < 0$ then allocate exactly $\lceil - y_i^{t} \rceil + 1 \geq 2$ balls to underloaded bins such that there can be at most two bins $k_1, k_2 \in B_-^t$ where:
		\begin{enumerate}
			\item $k_1 \in B_-^t$ receives $\lceil -y_{k_1}^t \rceil + 1$ balls (i.e., $x_{k_1}^{t+1}= \big\lceil \frac{W^{t}}{n} \big\rceil + 1$),
			\item $k_2 \in B_-^t$ receives $\lceil -y_{k_2}^t \rceil$ balls  (i.e., $x_{k_2}^{t+1}= \big\lceil \frac{W^{t}}{n} \big\rceil $),
			\item each $j \in B_-^t\backslash\{k_1,k_2\}$ receives at most $\lceil -y_{j}^t \rceil - 1$ balls (i.e.,  $x_{j}^{t+1} \leq \big\lceil \frac{W^{t}}{n} \big\rceil -1 $).
		\end{enumerate}
	\vspace{-.5cm}
	\end{itemize} 
\end{framed}

\begin{figure}[ht]
	\centering
	\includegraphics[scale=0.56]{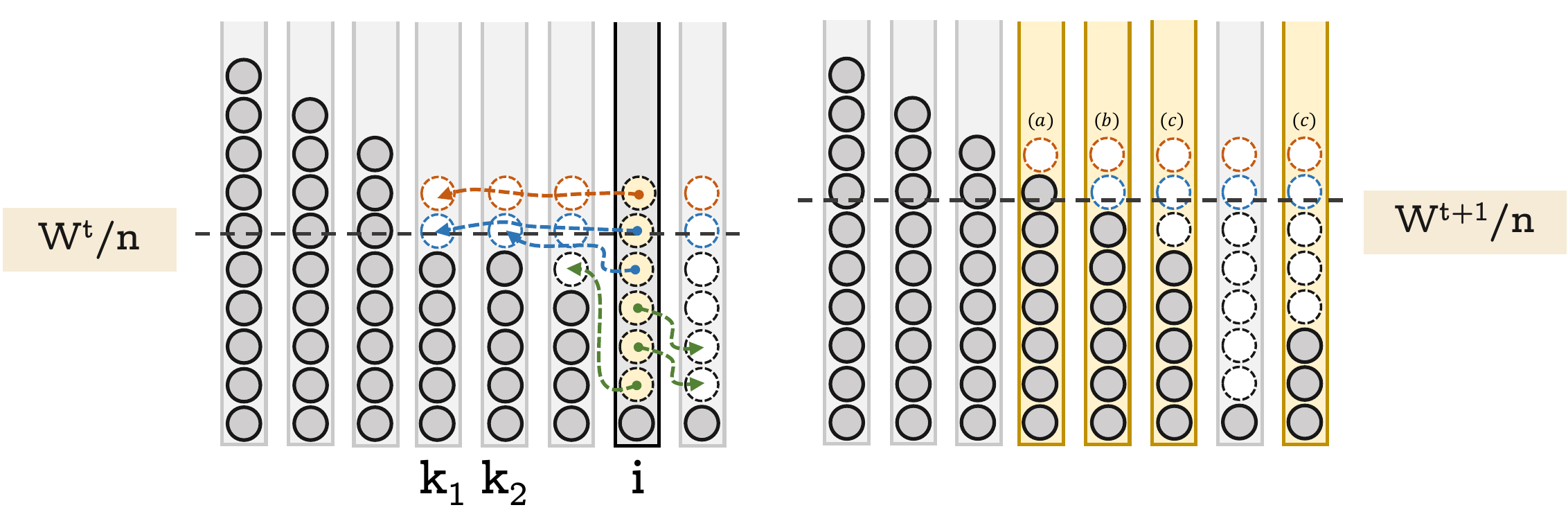}
	\caption{Illustration of one round for a \Filling process. After the underloaded bin $i$ is picked, $\lceil -y_i^{t} \rceil + 1 = 6$ balls are allocated. Only one bin $k_1$ receives $\lceil -y_{k_1}^t \rceil + 1$ balls and only one bin $k_2$ receives $\lceil -y_{k_2}^t \rceil$ balls, i.e., at most one allocated bin attains a load in the orange region and at most one in the blue region. All other bins $j$ receive at most $\lceil -y_{j}^t \rceil - 1$ balls.}
	\label{fig:general_filling_process}
\end{figure}

See \Cref{fig:general_filling_process} for an illustration of the effect of condition \WOne, if an underloaded bin is chosen.

We emphasize that in condition \POne, the distribution $p^{t}$ may not only be time-dependent, but may also depend on the entire history of the process until round $t$, that is, the filtration $\mathfrak{F}^{t}$. We point out that condition \WOne seems both a bit technical and arbitrary. For clarity, observe that within one round condition \WOne can make at most two formerly unloaded bins overloaded; one bin by at most two balls (case $(a)$) and one bin by at most one ball (case $(b)$). This could be further relaxed (and generalized) by allowing a constant number of bins to become overloaded by a constant number of balls. However, this would come at the cost of making the analysis more tedious, while the specific choice of condition \WOne already suffices to cover the processes defined later in this section. 

Also, the allocation used in \WOne may depend on the filtration $\mathfrak{F}^{t}$. Thus the framework also applies in the presence of an adaptive adversary, which directs all the $\lceil -y_i^t \rceil + 1$ balls to be allocated in round $t$ towards the ``most loaded'' underloaded bins (see \TightPacking in \Cref{sec:processes}). At the other end of the spectrum, there are natural processes which have a propensity to place these balls into ``less loaded'' underloaded bins. One specific, more complex example is the \Memory process, where due to the update of the cache after each single ball, the allocation is more skewed towards ``less loaded'' underloaded bins. However, as we shall discuss shortly, \Memory only satisfies \POne and \WOne after a suitable ``folding'' of rounds.

Our main result is that \Filling processes have gap of order $\log n$ with high probability.

\begin{restatable}{thm}{fillingkey}\label{thm:filling_key}
	There exists a constant $C>0$, such that for any allocation process satisfying conditions \POne and \WOne, and any round $m \geq 1$, we have
	\[
	\Pro{ \Gap(m) \leq C \log n} \geq 1-n^{-2}.
	\]
\end{restatable}
This result is proven in \Cref{sec:analysis_filling} by analyzing a variant of the exponential potential function $\Phi$ over the overloaded bins and eventually establishing that for any $m \geq 1$,
$
\Ex{ \Phi^m } = \Oh(n).
$
Then a simple application of Markov's inequality yields the desired result for the gap.

Recall that $W^t = \sum_{i\in[n]} x_i^t$ is the total number of balls allocated by round $t$. The \textit{throughput} of a \Filling process at round $t$ is defined as\[
\mu^t := \frac{W^t}{t},
\]
this is the average balls placed per bins sampled during the first $t$ rounds. The following theorem bounds the expected throughput for \Filling processes.

\begin{restatable}{thm}{sampleefficient}\label{thm:sample_efficient}
	There exist constants $c,C >0$ such that for any allocation process satisfying conditions \POne and \WOne and round $m\geq 2$, we have
	\[\ex{\mu^m }\geq 1 + c, \]
	and if additionally the probability vector of the allocation process is uniform at each round (this still satisfies \POne), then for any round $m \geq 2$,  	\[1 + c \leq \ex{\mu^m }\leq C.\]
\end{restatable}  

Further, recall that $S^t$ is the number of bins sampled by round $t$. We define the \textit{sample efficiency} by  \[ \eta^t  = \frac{W^t}{S^t},\] this is the average balls placed per bin sampled during the first $t$ rounds. Observe that for any $t\geq 1$ we have $\eta^t = 1$ for \OneChoice and $\eta^t = 1/2$ for \TwoChoice deterministically. For \Packing, $S^t = t$ and so $\eta^t = \mu^t$. Therefore, \Cref{thm:sample_efficient} implies that \Packing is more sample-efficient than \OneChoice by a constant factor in expectation. 
\begin{restatable}{cor}{PackingSampleEfficient} \label{cor:packing_sample_efficient}
	There exist constants $c, C > 0$ (given by \Cref{thm:sample_efficient}) such that for any round $m \geq 2$, the the \Packing process satisfies \[
	1 + c \leq \Ex{\eta^m} \leq C.
	\]
\end{restatable}

The key idea in the proof of \Cref{thm:sample_efficient} is the observation that within any window of $n$ rounds, in a constant proportion of the rounds there is either a large number of underloaded bins or the absolute value potential function is linear. If the first event holds in a round, then we pick an underloaded bin with constant probability, thus conditional on this event the expected number of balls increases by more than one since we allocate at least two balls to an underloaded bin. The increase in the expected number of balls allocated conditional on the second event comes from a direct relationship between the number of balls allocated in a round and the absolute value potential.    

For the lightly loaded case $m = n$, we can obtain a tighter upper bound on the gap, by observing that for a bin to have a gap of $g$ it must be chosen as an overloaded bin at least $g-2$ times. Hence, by the maximum load of \OneChoice for $m = n$ (e.g., \cite[Lemma 5.1]{MU2017}), we get:

\begin{obs} \label{obs:upper_bound_lightly_loaded}
	For any allocation process satisfying conditions \POne and \WOne we have
	\[
	\Pro{\Gap(n) \leq 3 \cdot \frac{\log n}{\log \log n}} \geq 1 - n^{-1}.
	\]
\end{obs}
In \Cref{lem:uniform_prob_lower_bound} we prove a matching lower bound for any process which samples one bin uniformly in each round. 
\subsection{The Packing Process} \label{sec:processes}

A natural example of a process satisfying \POne and \WOne is \Packing, which places ``greedily'' as many balls as possible into an underloaded bin (i.e., a bin with load below average).   

\begin{framed}
	\vspace{-.45em} \noindent
	\underline{\Packing~Process:}\\
	\textsf{Iteration:} For each $t \geq 0$, sample a uniform bin $i$, and update its load:
	\begin{equation*} x_{i}^{t+1} =
		\begin{cases}
			x_{i}^{t} + 1 & \mbox{if $x_{i}^{t} \geq  \frac{W^{t}}{n} $}, \\
			\left\lceil \frac{W^{t}}{n} \right\rceil + 1  & \mbox{if $x_{i}^{t} <  \frac{W^{t}}{n} $}.
		\end{cases}
	\end{equation*}\vspace{-1.5em}
\end{framed} 

See \Cref{fig:packing_process} for an illustration of \Packing. It is simple to show that this is a \Filling process.

\begin{figure}[h!]
	\begin{center}
		\includegraphics[scale=0.5]{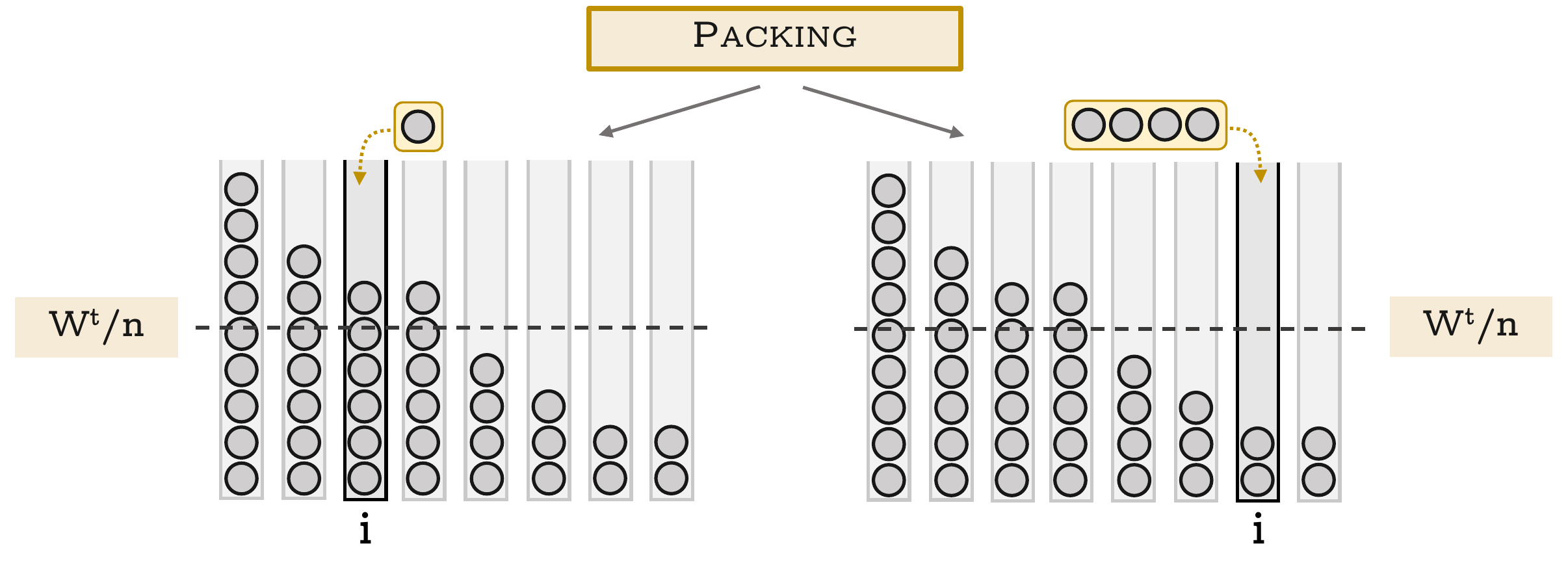}
	\end{center}
	\vspace{-0.3cm}
	\caption{Illustration of the two different possibilities in a single round of the \Packing process: (\textbf{left}) allocating a single ball to an overloaded bin and (\textbf{right}) filling an underloaded bin.}\label{fig:packing_process}
\end{figure}

\begin{lem}
	The \Packing process satisfies conditions \POne and \WOne.
\end{lem}
\begin{proof}
	The \Packing process picks a uniform bin $i$ at each round $t$, thus it satisfies \POne. Furthermore, if $y_i^{t} < 0$, it allocates exactly $\lceil -y_i^{t} \rceil + 1$ balls to bin $i$; otherwise, it allocates one ball to $i$, and thus \WOne is also satisfied.
\end{proof}

The \Packing process is quite similar to \OneChoice in that it samples one bin per round, however, Theorem \ref{thm:filling_key} shows  that it has a $\Oh(\log n)$ gap. We also prove the following lower bound which shows that the upper bound is essentially best possible.

\begin{restatable}{thm}{packlower} \label{lem:packing_logn_lb}
	There exists a constant $\kappa > 0$ such that for any $m \geq \kappa n \log n$ the \Packing process satisfies 	\[
	\Pro{\Gap(m) \geq \frac{\sqrt{\kappa}}{20} \cdot \log n} \geq \frac{1}{2}.
	\]
\end{restatable}

The \Packing process arises naturally when one is trying to achieve a small gap from few bin queries or random samples. In contrast, the following process is rather contrived, as whenever an underloaded bin is chosen, balls can be placed into (possibly different) underloaded bins that have the highest load. However, it is interesting that even this process achieves a small gap, as we shall show it is a \Filling process.  Furthermore, experiments suggest that this process frequently leads to load configurations where the lightest bin has a normalized load of $-\omega(\log n)$. We leave it as an open problem to establish this theoretically. If this is the case then it would show that the hyperbolic cosine potential with constant smoothing parameter used in \cite{PTW15} cannot be used directly to deduce an $\mathcal{O}(\log n)$ gap bound for \Filling processes.   
\begin{framed}
	\vspace{-.45em} \noindent
	\underline{\TightPacking~Process:}\\
	\textsf{Iteration:} For each $t \geq 0$, sample a uniform bin $i$, and update:\vspace{-.2cm }
	\begin{equation*}
		\left\{ \!\!\!	\begin{tabular}{ l l   }
			$x_{i}^{t+1} = x_{i}^{t} + 1$ &  \mbox{\phantom{shift}if $x_{i}^{t} \geq  \frac{W^{t}}{n} $}, \vspace{.25cm }\\
			\text{allocate $\lceil - y_i^t\rceil + 1$ balls one by one into the bins $\ell$} &  \multirow{3}{*}{\mbox{\phantom{shift}if $x_{i}^{t} <  \frac{W^{t}}{n}$,}}   \\ 
			\text{with highest loads such that $x_{\ell}^{t+1} < \frac{W^t}{n}$, except for} &   \\
			\text{one bin $j$ that gets $x_j^{t+1} = \big\lceil \frac{W^{t}}{n} \big\rceil + 1$.} &
		\end{tabular}\right. 
	\end{equation*}\vspace{-1.5em}
\end{framed}

An equivalent (and more formal) description of \TightPacking is the following: Assume the bin loads are decreasingly sorted  $x_1^{t} \geq x_2^{t} \geq \cdots \geq x_{n}^t$. If the selected bin $i\in [n]$ is overloaded, then we allocate one ball in $i$. Otherwise, we update the load of the maximally loaded underloaded bin $j \in [n]$ at round $t$ to $x_j^{t+1} = \big\lceil \frac{W^t}{n} \big\rceil + 1$. 
Then the remaining $(\lceil - y_i^t \rceil + 1) - (\lceil -y_j^t \rceil +1) = \lceil - y_i^t \rceil - \lceil -y_j^t \rceil \geq 0$ balls (if there are any), are allocated to bins $j+1,j+2,\ldots,j+\ell$ for some integer $\ell \geq 0$, such that all bins $k \in [j+1,j+\ell-1]$ have $x_k^{t+1} = \big\lceil \frac{W^t}{n} \big\rceil - 1$ and $x_{j + \ell}^{t+1} \in [ x_{j + \ell}^t, \big\lceil \frac{W^t}{n} \big\rceil )$. See \Cref{fig:tight_packing_process} for an illustration of the \TightPacking process. 
\begin{figure}[H]
	\begin{center}
		\includegraphics[scale=0.5]{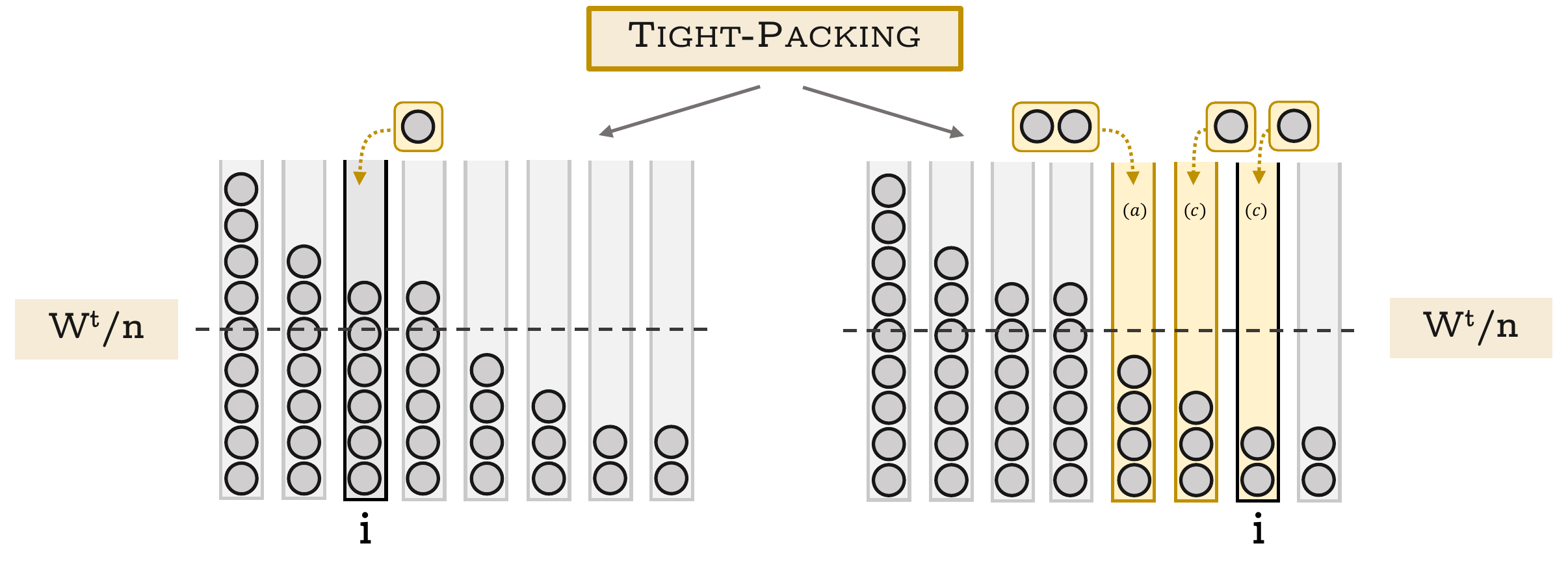}
	\end{center}
	\vspace{-0.3cm}
	\caption{Illustration of the two different possibilities in a single round of the \TightPacking process: (\textbf{left}) allocating a single ball to an overloaded bin and (\textbf{right}) allocating as many balls as the underload of the selected bin, to the most underloaded bins. Note that only one of the three bins where balls were allocated, attains a load of $\geq \lceil W^t/n\rceil$.}\label{fig:tight_packing_process}
\end{figure}

As promised we now show that this is also a \Filling process.

\begin{lem}
	The \TightPacking process satisfies conditions \POne and \WOne.
\end{lem}
\begin{proof}
	The \TightPacking process picks a uniform bin $i$ at each round $t$, thus it satisfies \POne.
	Further, the allocation satisfies the following properties regarding the allocation of balls.
	
	First, in case of $y_i^{t} < 0$, then:
	$(i)$ we allocate
	exactly $\lceil -y_i^{t} \rceil + 1$ balls, $(ii)$ one bin receives $\lceil - y_j^{t} \rceil + 1$ balls (satisfying $(a)$), 
	and $(iii)$ every other bin $j$ receives a number of balls between $[0, \lceil -y_j^{t} \rceil - 1]$ (satisfying $(c)$).
	
	Secondly, in case of $y_i^{t} \geq 0$ we place one ball to $i$. Thus \WOne is also satisfied.
\end{proof}

\subsection{Memory and Unfolding}\label{sec:memory}

The \Memory process was introduced by Mitzenmacher, Prabhakar \& Shah~\cite{MPS02} and works under the assumption that the address $b$ of a single bin can be stored or ``cached''. The process is essentially a \TwoChoice process where the second sample is replaced by the bin in the cache. 

\begin{framed}
	\vspace{-.45em} \noindent
	\underline{\Memory Process:}\\
	\textsf{Iteration:} For each $t \geq 0$, sample a uniform bin $i$, and update its load (or of cached bin $b$):\vspace{-0.3cm}
	\begin{equation*}
		\begin{cases}
			x_{i}^{t+1} = x_{i}^{t} + 1  & \mbox{if $x_{i}^{t} <  x_{b}^{t} $} \qquad \mbox{(also update cache $b=i$)}, \\
			x_{i}^{t+1} = x_{i}^{t} + 1  & \mbox{if $x_{i}^{t} =  x_{b}^{t} $}, \\
			x_{b}^{t+1} = x_{b}^{t} + 1 & \mbox{if $x_{i}^{t} >  x_{b}^{t}$}.
		\end{cases}
	\end{equation*}\vspace{-1.5em}
\end{framed}

See \Cref{fig:memory_process} for an illustration of the \Memory process.

As mentioned above, \Memory does not satisfy conditions \POne and \WOne directly. The issue is that \Memory allocates only {\em one} ball at each round whereas, due to condition \WOne, a \Filling process may place several balls into underloaded bins in a single round. 
To overcome this issue, we define a so-called \textit{unfolding} of a \Filling process. The \Filling process proceeds in rounds $0,1,\ldots$, whereas the unfolding is a coupled process, which is encoded as a sequence of \emph{atomic allocations} (each allocation places exactly one ball into a bin).
Using this unfolding along with the gap bound for a \Filling process satisfying \POne and \WOne, we can then bound the number of atomic allocations with a large gap.

\begin{figure}[H]
	\begin{center}
		\includegraphics[scale=0.5]{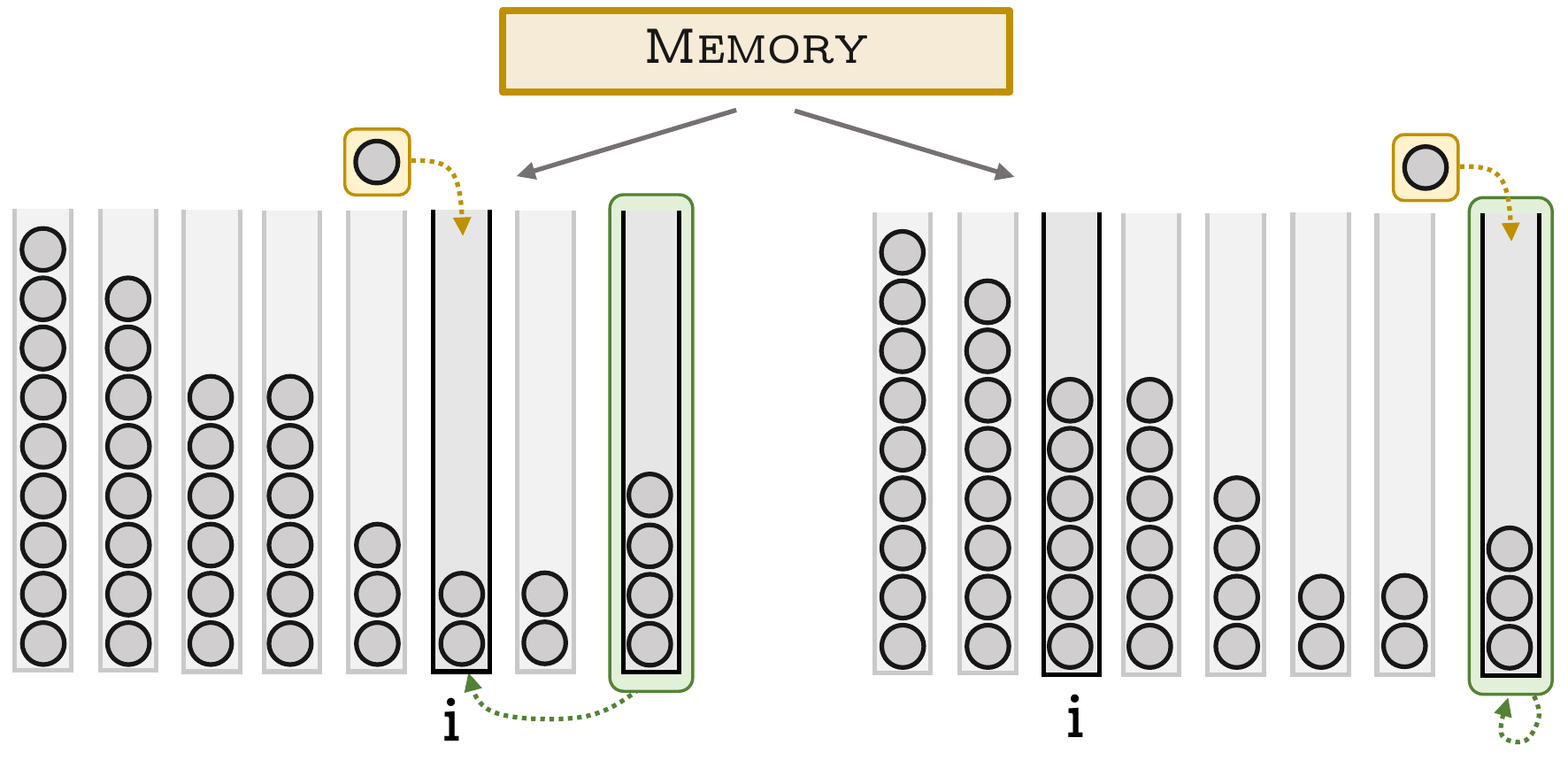}
	\end{center}
	\vspace{-0.3cm}
	\caption{Illustration of the two different possibilities in a single round of the \Memory process: (\textbf{left}) allocating to the sampled bin and (\textbf{right}) updating the cache or allocating to the cache (shown in green).}\label{fig:memory_process}
\end{figure}

First, we define unfolding formally:\label{folding} \def\unfolding{Let $x^{t}$ be the load vector of a \Filling process $P$, and $\widehat{x}^{\,t}$ be the load vector of the unfolding of $P$ called $U=U(P)$. Initialize $x^{0}:=\widehat{x}^{\,0}$ as the all-zero vector, corresponding to the empty load configuration. For every round $t \geq 1$, where $P$ has allocated $W^{t-1}$ balls in the previous rounds, we create atomic allocations  $A(t):=\{W^{t-1}+1,\ldots,W^{t}\}$ for the process $U$, where $W^{t} = W^{t-1}+\min\left\{\lceil -y_i^t \rceil + 1, \; 1\right\}$, such that for each $s \in A(t)$ a single ball is allocated in the process $U$.
	Additionally, we have $x^{t}:=\widehat{x}^{\,W^{t}}$, i.e., the load distribution of $P$ after round $t$ corresponds to load distribution of $U$ after atomic allocation $W^{t} \geq t$.}\unfolding\ Note that the unfolding of a process is not (completely) unique, as the allocations in $[W^{t-1}+1,W^{t}]$ of $U(P)$ can be permuted arbitrarily.

The next result shows that \Memory can be considered as a filling process after it has been ``unfolded''.

\begin{restatable}{lem}{memorycompressed}\label{lem:memory_compressed}
	There is an allocation process $P$ satisfying conditions \POne and \WOne, such that for a suitable unfolding $U=U(P)$, the allocation process $U$ is an instance of \Memory.
\end{restatable}

Now, applying \Cref{thm:filling_key} to the unfolding of a \Filling process, and exploiting that in a balanced load configuration not too many atomic allocations can be created through unfolding, we obtain the following result:

\def\sloweddown{
	For any constant $c> 0$ there exists a constant $C>0$ such that for any allocation process $P$ satisfying conditions \POne and \WOne, and for any number of atomic allocations $m \geq 1$, with probability at least $1-n^{-2}$, any unfolding $U=U(P)$ satisfies
	\[
	\bigl| \left\{ t \in [m] \colon \Gap_{U}(t) \geq C \cdot \log n \right\} \bigr| \leq n^{-c} \cdot m \log m.
	\]}

\begin{restatable}{lem}{sloweddown}\label{lem:slowed-down} 
	Fix any constant $c> 0$. Then for any allocation process $P$ satisfying conditions \POne and \WOne, there is a constant $C=C(c)>0$ such that for any number of atomic allocations $m \geq 1$, with probability at least $1-n^{-2}$, any unfolding $U=U(P)$ satisfies
	\[
	\bigl| \left\{ t \in [m] \colon \Gap_{U}(t) \geq C \cdot \log n \right\} \bigr| \leq n^{-c} \cdot m \log m.
	\]
\end{restatable}

Hence for any $m$, all but a polynomially small fraction of the first $m$ atomic allocations in any unfolding of a \Filling process (e.g., \Memory) have a logarithmic gap. This behavior matches the one of the original \Filling process, under the limitation that we cannot prove a small gap that holds for an arbitrarily large {\em fixed} atomic allocation $m$. However, if $m$ is polynomial in $n$, that is, $m \leq  n^{c}$ for some constant $c>0$, then \Cref{lem:slowed-down} implies directly that with high probability, the gap at atomic allocation $m$ (and at all atomic allocations before) is logarithmic.

We note however that for the specific case of \Memory, a $\mathcal{O}(\log n)$ gap bound is not tight. Since this paper was submitted, the authors proved a $\mathcal{O}(\log\log n)$ gap bound for the \Memory process which holds \Whp\ for any number $m\geq 1$ of balls \cite{SODA23}.

\subsection{The Packing Process with a Biased Probability Vector} \label{sec:packing_with_negative_bias}

Finally, we also consider the \Packing process with an $(a, b)$-biased sampling vector $p$, meaning that the probability of sampling bin $i$ satisfies $\frac{1}{an} \leq p_i \leq \frac{b}{n}$. Such a vector could even have a bias to place towards overloaded bins. For \Packing the sampling vector coincides with the probability vector, so we will refer to it as probability vector from this point forward.

\begin{framed}
	\vspace{-.45em} \noindent
	\underline{\Packing~process with an $(a, b)$-biased probability vector $p^t$:}\\
	\textsf{Parameter:} At any step $t$, probability vector $p^t := p(\mathfrak{F}^t)$ satisfying $\frac{1}{an} \leq p_i^t \leq \frac{b}{n}$ for all $i \in  [n]$. \\
	\textsf{Iteration:} For each $t \geq 0$, sample a bin $i \in_p [n]$, and update its load:
	\begin{equation*} x_{i}^{t+1} =
		\begin{cases}
			x_{i}^{t} + 1 & \mbox{if $x_{i}^{t} \geq  \frac{W^{t}}{n} $}, \\
			\Big\lceil \frac{W^{t}}{n} \Big\rceil + 1  & \mbox{if $x_{i}^{t} <  \frac{W^{t}}{n} $}.
		\end{cases}
	\end{equation*}\vspace{-1.em}
\end{framed}

Note that this process is not necessarily a \Filling process as its probability vector may not be majorized by the uniform distribution. Such probability vectors may occur, for instance, if the load information is noisy or outdated. Next we show that as long as $a, b$ are functions of $n$, the \Packing process has a gap bound that is independent of $m$. 
\begin{restatable}{pro}{newnegbias}\label{Pro:negbias}
	
	Consider the \Packing process with any $(a, b)$-biased probability vector with $a := a(n) > 1$ and $b := b(n) > 1$. Then, for any round $m \geq 1$, we have
	\[
	\Ex{\Delta^m} \leq 6an^3,
	\]
	and hence,
	\[
	\Pro{\Gap(m) \leq 6an^4} \geq 1 - n^{-1}.
	\]
\end{restatable}
This result is very simple to prove and we do not claim that the bound is tight (as it only makes use of the lower bound on the $p_i$'s), however we include it as it illustrates the \textit{power of filling}. In particular, in~\cite{W07} it was shown that for \TwoChoice for sufficiently large constants $a, b>1$ (taking $a,b=10$ suffices), there exist $(a, b)$-biased sampling distributions, for which the gap will grow at least as fast as \OneChoice, that is $\Gap(m) = \Omega(\sqrt{\frac{m}{n}\log n })$ \Whp\ for $m\gg n$. However, the \Packing process has also the power to ``fill'' underloaded bins, and it is this power that allows it to have a bounded gap -- despite the best efforts of its unruly probability vector.

\section{Upper Bounds on the Gap}\label{sec:analysis_filling}

In this section we prove our upper bounds on the gap. We begin in \Cref{sec:filling_potential} by defining the exponential potential function and analyzing its behavior (depending on some other constraints). Then in \Cref{sec:complete_filling_proof} we use a super-martingale argument to show that the exponential potential function decreases, which eventually yields the desired gap bound in \Cref{thm:filling_key}. Finally in \Cref{sec:biasneg} we bound the gap of the \Packing with an $(a, b)$-biased probability vector, by considering the absolute value potential.

\subsection{Potential Function Analysis}\label{sec:filling_potential}
We consider a version of the \textit{exponential potential function} $\Phi$ which only takes bins into account whose load is at least two above the average load. This is defined for any round $t \geq 0$ by

\[ 
\Phi^t := \Phi^t(\alpha) := \sum_{i: y_{i}^t\geq 2} \exp\left( \alpha \cdot  y_i^t \right) =  \sum_{i=1}^n \exp\left( \alpha \cdot  y_i^t \right)\cdot  \mathbf{1}_{\{ y_{i}^{t}\geq 2\}} ,
\]
where we recall that $y_i^t = x_i^t - \frac{W^t}{n}$ is the normalized load of  bin $i$ at round t and $\alpha > 0$ is a sufficiently small constant to be fixed later. Let $\Phi^t_i= \exp\left( \alpha \cdot  y_i^t \right)\cdot  \mathbf{1}_{\{ y_{i}^{t}\geq 2\}}$ and thus $\Phi^t=\sum_{i=1}^n\Phi^t_i$. 
We will also use the \textit{absolute value potential}: for any round $t \geq 0$,
\[
\Delta^{t}:=\sum_{i=1}^{n} \left| y_i^t \right|.
\]

The next lemma provides a useful upper bound on the expected change of $\Phi$ over one round. It establishes that to bound $\ex{\Phi^{t + 1} \mid \mathfrak{F}^t}$ from above, we may assume the probability vector $p^{t}$ is uniform.

\begin{lem}\label{lem:filling}Consider any allocation process satisfying conditions \POne and \WOne. Then, for $\Phi := \Phi(\alpha)$ with any $\alpha > 0$ and for any round $t \geq 0$, 
	\[
	\ex{\Phi^{t + 1} \mid \mathfrak{F}^t} \leq \frac{1}{n}  \sum_{i=1}^n \Phi_i^{t} \cdot\left(  \sum_{j \colon y_j^{t} < 1}  
	e^{\frac{-\alpha (\lceil -y_j^t \rceil + 1)}{n}}  +  e^{-\frac{\alpha}{n}} \cdot (|B_{\geq 1}^{t}|-1) +  e^{\alpha-\frac{\alpha}{n}} \right) + e^{3\alpha},
	\]
	where $B_{\geq 1}^t$ denotes the set of bins with load at least $1$. 
\end{lem}

\begin{proof}Recall that the filtration $\mathfrak{F}^t$ reveals the load vector $x^t$. Throughout this proof, we consider the labeling chosen by the process in round $t$ such that $x_1^{t} \geq x_2^{t} \geq \cdots \geq x_n^{t}$ and $p^t$ being majorized by \OneChoice (according to \POne). We emphasize that for this labeling, it is  quite possible that $x_i^{t+1}$ will not be non-increasing in $i \in [n]$.

	To begin, using $\Phi_i^{t+1}=e^{\alpha y_{i}^{t+1}} \cdot \mathbf{1}_{\{ y_{i}^{t+1}\geq 2\}}$, we can split $\Phi^{t+1}$ over the $n$ bins as follows,
	\[
	\Ex{ \Phi^{t+1} \, \mid \, \mathfrak{F}^t}   = \sum_{j=1}^n \Ex{
		\Phi_{j}^{t+1}
		\, \mid \, \mathfrak{F}^t }.
	\]
	Now consider the effect of picking bin $i$ for the allocation in round $t$ to the potential $\Phi^{t+1}$. Note that bin $i$ is chosen with probability equal to $p_i^t$. Observe that if bin $i$ satisfies $y_i^t<1$, then it receives at most $\lceil -y_i^t\rceil +1$ balls in round $t$, thus $y_{i}^{t+1}\geq 2$ if and only if $y_i^{t}\geq 1$. 
	Using condition \WOne, we distinguish between the following three cases based on how allocating to $i$ changes $\Phi_{j}^{t}$ for $j \neq i$ and for $j=i$:
	\medskip 
	
	\noindent\textbf{Case 1.A} [$y_i^{t} < 1$, $j\neq i$, $y_j^t < 1$].
	We will allocate $\lceil - y_i^{t} \rceil +1$ many balls to bins $k$ with $y_k^t <1$ (not necessarily to $i$) subject to \WOne. This increases the average load by $(\lceil - y_i^{t} \rceil + 1)/n$. Since $y_j^t < 1$, $\Phi_j^t =0 $. Further, by condition \WOne we can increase the load of $y_j^t$ by at most $\lceil - y_j^t \rceil + 1$, hence, $\Phi_j^{t+1}=0$. Therefore,
	$
	\Phi_{j}^{t+1} =  \Phi_j^{t} \cdot e^{\frac{-\alpha (\lceil -y_i^t \rceil + 1)}{n}}$ for $j \neq i$. 
	
	\medskip
	
	\noindent\textbf{Case 1.B} [$y_i^{t} < 1$, $j\neq i$, $y_j^t \geq 1$].
	As in Case 1a, we will allocate $\lceil - y_i^{t} \rceil +1$ many balls to bins $k$ with $y_k^t <1$ (not necessarily to $i$) subject to \WOne, which increases the average load by $(\lceil - y_i^{t} \rceil + 1)/n$.
	Additionally, bin $j$ will receive no balls (by condition \WOne), thus $\Phi_{j}^{t+1} =  \Phi_j^{t} \cdot e^{\frac{-\alpha (\lceil -y_i^t \rceil + 1)}{n}}$ for $j \neq i$. 
	
	\medskip 
	
	\noindent\textbf{Case 2} [$y_i^{t} \geq 1$, $j\neq i$]. We allocate one ball to $i$, which increases the average load by $1/n$, and thus $\Phi_{j}^{t+1} =  \Phi_j^{t} \cdot e^{-\frac{\alpha}{n}}$ for $j \neq i$, which again also holds for bins $j$ that do not contribute. 
	
	\medskip 
	
	\noindent\textbf{Case 3} [$j=i$].
	Finally, we consider the effect on $\Phi_{i}^{t+1}$. Again if $y_i^{t} < 1$, then $\Phi_{i}^t=0$ and $\Phi_{i}^{t+1}=0$. Otherwise, we have $y_i^{t} \geq 1$ and we allocate one ball to $i$, and thus
	\[
	\Phi_{i}^{t+1} =  e^{ \alpha \cdot (y_{i}^{t}+
		1) - \frac{\alpha}{n}}\mathbf{1}_{\{ y_{i}^{t+1}\geq 2\}} = e^{ \alpha \cdot y_i^{t}}\mathbf{1}_{\{ y_{i}^{t+1}\geq 2\}} \cdot e^{ \alpha - \frac{\alpha}{n}} \leq 
	(e^{ \alpha \cdot y_i^{t}}\mathbf{1}_{\{ y_{i}^{t}\geq 2\}} +e^{2\alpha} ) \cdot e^{ \alpha - \frac{\alpha}{n}},
	\]
	where the $+e^{2\alpha}$ is added to account for the case where $1\leq y_i^{t} <2 $ and so $\Phi_{i}^t =0$ but $0<\Phi_{i}^{t+1} \leq e^{3\alpha -\alpha/n}$. Thus in this case, $\Phi_{i}^{t+1}\leq \Phi_{i}^{t}e^{ \alpha - \frac{\alpha}{n}} +  e^{3 \alpha - \frac{\alpha}{n}}\leq \Phi_{i}^{t}e^{ \alpha - \frac{\alpha}{n}} +  e^{3 \alpha }.$
	\medskip 
	
	By aggregating the three cases above, and observing that $\sum_{i=1}^n p_i^t\cdot e^{3 \alpha }=e^{3 \alpha }$ , we see that 
	\begin{align}
		\sum_{j=1}^n \ex{\Phi_{j}^{t + 1} \mid \mathfrak{F}^t}    &=
		\sum_{i=1}^n p_i^t \sum_{j=1}^n \ex{\Phi_{j}^{t + 1} \mid \mathfrak{F}^t, \;\text{Bin $i$ is selected at round $t$}}   \notag  \\  &\leq \sum_{i=1}^n  p_i^t \cdot \mathbf{1}_{\{y_i^t < 1\}}\sum_{j=1}^n \Phi_j^{t}  \cdot e^{\frac{-\alpha (\lceil -y_i^t \rceil + 1)}{n}}  + \sum_{i =1}^n  p_i^t \cdot \mathbf{1}_{\{y_i^t \geq 1\}}\sum_{j \neq i} \Phi_j^{t} \cdot e^{-\frac{\alpha}{n}} \notag \\ &\qquad + \sum_{i=1}^n p_i^t
		\cdot  \mathbf{1}_{\{y_i^t \geq 1\}}\cdot  \Phi_i^{t} \cdot e^{\alpha-\frac{\alpha}{n}}+ e^{3 \alpha }. \label{eq:firstbd} 
	\end{align}
	
	We will now rewrite \eqref{eq:firstbd} in order to establish that it is maximized if $p$ is the uniform distribution.
	Adding $\sum_{i=1}^n p_i^{t} \cdot \mathbf{1}_{\{y_i^t \geq 1\}}\cdot \Phi_i^{t} \cdot e^{-\frac{\alpha}{n}}$ to the middle sum (corresponding to Case 2) in \eqref{eq:firstbd} and subtracting it from the last sum (corresponding to Case 3) transforms \eqref{eq:firstbd} into
	\begin{align*}
		&\sum_{i=1}^n  p_i^t  \mathbf{1}_{\{y_i^t < 1\}}\sum_{j=1}^n \Phi_j^{t}   e^{\frac{-\alpha (\lceil -y_i^t \rceil + 1)}{n}}  + \sum_{i =1}^n  p_i^t  \mathbf{1}_{\{y_i^t \geq 1\}}\sum_{j=1}^n \Phi_j^{t}  e^{-\frac{\alpha}{n}}  + \sum_{i=1}^n p_i^t
		\Phi_i^{t}  \mathbf{1}_{\{y_i^t \geq 1\}}e^{ -\frac{\alpha}{n}}\left(e^{\alpha}-1\right) + e^{3 \alpha}\notag  \\
		& = \sum_{i=1}^n  p_i^t \!\left( \! \! \Bigg(\underbrace{ \mathbf{1}_{\{y_i^t < 1\}}   e^{\frac{-\alpha (\lceil -y_i^t \rceil + 1)}{n}} + \mathbf{1}_{\{y_i^t \geq 1\}}\cdot e^{-\frac{\alpha}{n}} }_{g(i)} \Bigg) \!\sum_{j=1}^n  \Phi_j^{t} +   \underbrace{ \Phi_i^{t} \cdot \mathbf{1}_{\{y_i^t \geq 1\}}e^{ -\frac{\alpha}{n}}\left(e^{\alpha}-1\right) }_{f(i)} \right)  + e^{3 \alpha } .
	\end{align*}

	Recall that $y_1^{t} \geq y_2^{t} \geq \cdots \geq y_n^{t}$, which implies $\Phi_1^{t} \geq \Phi_2^{t} \geq \cdots \geq \Phi_n^{t}\geq 0$. Thus $f(i)$ and $g(i)$ are non-negative and non-increasing in $i$ and $\sum_{j=1}^{n} \Phi_j^{t}\geq 0 $. Consequently, the function $h(i) = f(i)+ g(i)\cdot \sum_{j=1}^n \Phi_j^{t}$ is non-negative and non-increasing in $i$. Note that by condition \POne, for any $k \in [n]$ it holds that $\sum_{i=1}^{k} p_i^t \leq \frac{k}{n}$. Thus we can apply \Cref{lem:quasilem} which implies $ \sum_{i=1}^n p_i^t \cdot h(i) \leq \sum_{i=1}^n \frac{1}{n} \cdot h(i)$. Applying this to the above, rearranging, and splitting $f(i)$ gives 
	\begin{equation}\begin{aligned}\label{eq:nearlythere!}
			\ex{\Phi^{t + 1} \mid \mathfrak{F}^t}&\leq \frac{1}{n} \sum_{i=1}^n   \mathbf{1}_{\{y_i^t < 1\}}   e^{\frac{-\alpha (\lceil -y_i^t \rceil + 1)}{n}} \sum_{j=1}^n  \Phi_j^{t}  +  \frac{1}{n} \sum_{i=1}^n  \mathbf{1}_{\{y_i^t \geq 1\}}\cdot e^{-\frac{\alpha}{n}}  \sum_{j=1}^n  \Phi_j^{t} \\&\qquad -\frac{1}{n}\sum_{i=1}^n  \mathbf{1}_{\{y_i^t \geq 1\}} \Phi_i^{t} \cdot e^{ -\frac{\alpha}{n}} +    \frac{1}{n}\sum_{i=1}^n   \mathbf{1}_{\{y_i^t \geq 1\}}\Phi_i^{t} \cdot e^{\alpha -\frac{\alpha}{n}}   + e^{3 \alpha }
	\end{aligned}\end{equation} Now observe that combining the second and third terms above gives
	\begin{align}&\frac{1}{n} \sum_{i=1}^n  \mathbf{1}_{\{y_i^t \geq 1\}}\cdot e^{-\frac{\alpha}{n}}  \sum_{j=1}^n  \Phi_j^{t} -\frac{1}{n}\sum_{i=1}^n  \mathbf{1}_{\{y_i^t \geq 1\}} \Phi_i^{t} \cdot e^{ -\frac{\alpha}{n}} \notag \\ 
		&\qquad = \frac{1}{n}\cdot e^{-\frac{\alpha}{n}}\cdot \sum_{i=1}^n  \sum_{j=1}^n  \mathbf{1}_{\{y_i^t \geq 1\}}\cdot  \Phi_j^{t} \cdot \mathbf{1}_{\{j\neq i\}} \notag \\
		&\qquad =   \frac{1}{n}\cdot e^{-\frac{\alpha}{n}}\cdot  \sum_{i=1}^n \Phi_i^{t} \sum_{j=1}^n  \mathbf{1}_{\{y_j^t \geq 1\}}\cdot \mathbf{1}_{\{j\neq i\}} \notag \\
		&\qquad =  \frac{1}{n}\cdot e^{-\frac{\alpha}{n}}\cdot  \sum_{i=1}^n \Phi_i^{t} \cdot (|B_{\geq 1}^{t}|-1), \label{eq:reshuffle} \end{align} where the last line follows since $ \Phi_i^{t}  = e^{\alpha y_i^t} \cdot \mathbf{1}_{\{y_{i}^t \geq 2\}}$. 
	
	Now, substituting \eqref{eq:reshuffle} into \eqref{eq:nearlythere!}, exchanging the first double summation and using the bound $\mathbf{1}_{\{y_i^t\geq 1 \}}\leq 1 $ on the last sum, and finally grouping terms gives the following  
	\begin{equation*} 
		\ex{\Phi^{t + 1} \mid \mathfrak{F}^t} \leq \frac{1}{n}  \sum_{i=1}^n  \Phi_i^{t} \left( \sum_{j=1}^n\mathbf{1}_{\{y_j^t < 1\}}  e^{\frac{-\alpha (\lceil -y_j^t \rceil + 1)}{n}} +  (|B_{\geq 1}^{t}|-1)e^{ -\frac{\alpha}{n}} +     e^{\alpha -\frac{\alpha}{n}} \right)   + e^{3 \alpha },\end{equation*}  
	as claimed.  
\end{proof}

Let $\mathcal{G}^t$ be the event that at round $t\geq 0$ either there are at least $n/20$ underloaded bins or there is a an absolute value potential of at least $n/10$. In symbols this is given by
\begin{equation}\label{eq:eventGt}
	\mathcal{G}^t := \left\{B_{-}^{t} \geq n/20 \right\}\cup \left\{ \Delta^t \geq n/10\right\}. 
\end{equation}
The next lemma provides two estimates on the expected exponential potential $\Phi^{t+1}$ in terms of $\Phi^{t}$. The first estimate holds for any round and it establishes that the process does not perform worse than \OneChoice, meaning that the potential increases by a factor of at most $(1+\Oh(\alpha^2/n))$. The second estimate is stronger for rounds where we have a lot of underloaded bins or a large value of the absolute value potential. This stronger estimate states that the potential decreases by a factor of $(1-\Omega(\alpha/n))$ in those rounds. Note that as we show in~\Cref{clm:counterexample}, the potential may increase in expectation for certain load configurations, so it seems hard to prove a decrease without looking at several rounds.

\begin{lem}
	\label{lem:filling_good_quantile}
	There exist constants $c_1,c_2 >0$ such that for any  $0 < \alpha < 1/100$, if $\Phi := \Phi(\alpha)$ is the potential of any allocation process satisfying conditions \POne and \WOne, then for any round $t \geq 0$
	\begin{align*}
		(i) & \qquad \Ex{ \Phi^{t+1} \, \mid \,\mathfrak{F}^t} \leq  \left( 1 
		+ \frac{c_1 \alpha^2}{n} \right) \cdot \Phi^{t} 
		+ e^{3 \alpha}, \\
		(ii) & \qquad \Ex{ \Phi^{t+1} \, \mid \, \mathfrak{F}^t, \mathcal{G}^t } \leq \left(1 - \frac{c_2 \alpha}{n} \right) \cdot \Phi^{t}  + e^{3\alpha}.
	\end{align*}
\end{lem}

\begin{proof}Recall that, as before, we fix the labeling chosen by the process in round $t$ such that $x_1^{t} \geq x_2^{t} \geq \cdots \geq x_n^{t}$, thus it is possible that $x_i^{t+1}$ may not be non-increasing in $i \in [n]$ and $p^t$ being majorized by \OneChoice (according to \POne). 
	
	Let $A_i$ be the bracketed term in the expression for $\ex{\Phi^{t + 1} \mid \mathfrak{F}^t} $ in \Cref{lem:filling}, given by 
	\begin{equation}\label{eq:Ai}A_i=   \sum_{j \colon   y_j^{t} < 1}
		e^{\frac{-\alpha (\lceil -y_j^t \rceil + 1)}{n}}  +  (|B_{\geq 1}^{t}| -1)\cdot e^{-\frac{\alpha}{n}} +   e^{\alpha-\frac{\alpha}{n}}.\end{equation}
	Observe that $-\alpha (\lceil -y_j^t \rceil + 1) \leq -\alpha$ whenever $y_j^{t} < 1$ and thus 
	\[A_i \leq    \sum_{j \colon   y_j^{t} < 1}
	e^{-\frac{\alpha}{n}}  +  (|B_{\geq 1}^{t}| -1)\cdot e^{-\frac{\alpha}{n}} +   e^{\alpha-\frac{\alpha}{n}} = e^{-\frac{\alpha}{n}} \cdot \left[ n-1 + e^{\alpha} \right].
	\]Applying the Taylor estimate $e^{z} \leq 1+z+z^2$, which holds for any $z \leq 1$, twice gives 
	\begin{align*}A_i&\leq  \left(1 - \frac{\alpha}{n} +  \frac{\alpha^2}{n^2}  \right) \left( n + \alpha + \alpha^2 \right) \\ 
		&= n \cdot \left(1 - \frac{\alpha}{n} +  \frac{\alpha^2}{n^2} \right) \left( 1 + \frac{\alpha}{n} +  \frac{\alpha^2}{n} \right) \\ 
		&\leq n \cdot \left(1 + \frac{c_1\alpha^2}{n}  \right),
	\end{align*}
	for some constant $c_1 >0$. The first statement in the lemma now follows as \Cref{lem:filling} gives 
	\begin{align*}
		\ex{\Phi^{t + 1} \mid \mathfrak{F}^t  } &\leq \frac{1}{n}  \sum_{i=1}^n \Phi_i^{t} \cdot A_i + e^{3\alpha} \\
		&\leq \frac{1}{n} \cdot n\left( 1 + \frac{c_1 \alpha^2}{n}  \right)\cdot  \sum_{i=1}^n \Phi_i^{t} + e^{3\alpha}\\ 
		&\leq\left( 1 +\frac{c_1 \alpha^2}{n} \right)\Phi^{t} + e^{3\alpha}. \end{align*} 
	We shall now show the second statement of the lemma. By splitting sums in \eqref{eq:Ai} we have 
	\begin{align} 
		A_i&=  \sum_{j \in B_{-}^t}
		e^{\frac{-\alpha (\lceil -y_j^t \rceil + 1)}{n}} + \sum_{j \colon 0\leq y_j^{t} < 1}
		e^{\frac{-\alpha (\lceil -y_j^t \rceil + 1)}{n}}  +  (|B_{\geq 1}^{t}| -1)\cdot e^{-\frac{\alpha}{n}} +   e^{\alpha-\frac{\alpha}{n}}\notag \\
		&=  \sum_{j \in B_{-}^t}
		e^{\frac{-\alpha (\lceil -y_j^t \rceil + 1)}{n}} +(|B_{+}^{t}|-1)\cdot e^{-\frac{\alpha}{n}}+   e^{\alpha-\frac{\alpha}{n}}\label{eq:A_ibdd}\\
		&\leq |B_{-}^t|\cdot e^{-\frac{2\alpha}{n}} +  (|B_{+}^{t}|-1)\cdot e^{-\frac{\alpha}{n}} +    e^{\alpha-\frac{\alpha}{n} }.\label{eq:A_ibdd2} \end{align}
	We shall now first assume that $|B_-^t| \geq n/20$. Recall the bound $e^x\leq 1 + x + 0.6 \cdot x^2 $ which holds for any $x\leq 1/2$. If $\alpha<1/2$, we can apply this bound to \eqref{eq:A_ibdd2}, giving
	\begin{align} 
		A_i&\leq e^{-\frac{\alpha}{n}}\cdot \left(|B_{-}^t|\cdot \left(1 -\frac{\alpha}{n} + \frac{6\alpha^2}{10n^2}\right) +  (|B_{+}^{t}|-1)+    1 +\alpha  + \frac{6\alpha^2}{10} \right) \notag  \\
		&\leq  \left(1 -\frac{\alpha}{n} + \frac{6\alpha^2}{10n^2}\right)\cdot n\left( 1 -\frac{\alpha}{20n} + \frac{6\alpha^2}{200n^2} +   \frac{\alpha}{n}  + \frac{6\alpha^2}{10n} \right)  \notag\\
		&= n\cdot\left(1 -\frac{\alpha(1-12\alpha)}{20n} + \Oh\!\left(\frac{\alpha^2}{n^2} \right)  \right). \label{eq:1stcondredux} \end{align}  
	We now assume that $|\Delta^{t}| \geq n/10$. 
	Observe that by Schur-convexity (see~\Cref{clm:schur}) and the assumption on $|\Delta^{t}|$ we have 
	\begin{align*}
		\sum_{j \in B_{-}^t}
		e^{\frac{-\alpha (\lceil -y_j^t \rceil + 1)}{n}} &\leq e^{-\frac{\alpha}{n}} \sum_{j \in B_{-}^t}e^{\frac{\alpha  y_j^t }{n}} \\ &\leq  e^{-\frac{\alpha}{n}} \cdot \left( (|B_{-}^{t}|-1) \cdot e^{-\frac{\alpha}{n}\cdot 0 } + 1 \cdot e^{\frac{\alpha}{n} \cdot \sum_{j \in B_{-}^t} y_i^{t} } \right) \\ &\leq (|B_{-}^{t}|-1) \cdot e^{-\frac{\alpha}{n}} +  e^{-\alpha/20}, \end{align*}
	where we used the fact that $\sum_{j \in B_{-}^t} y_i^{t}= -\frac{1}{2} \Delta^t$.
	Applying this and the bound $e^x\leq 1 + x + 0.6 \cdot x^2 $, for $x\leq 1/2$, to \eqref{eq:A_ibdd} gives  
	\begin{align}
		A_i & \leq  (|B_{-}^{t}|-1) \cdot e^{-\frac{\alpha}{n}} +  e^{-\alpha/20}   +   (|B_{+}^{t}| -1)\cdot e^{-\frac{\alpha}{n}}+   e^{\alpha}\notag \\
		& = (n-2) \cdot e^{-\frac{\alpha}{n}} +  e^{-\alpha/20}  +   e^{\alpha}\notag \\
		&\leq (n-2)\cdot \left(1 -\frac{\alpha}{n} + \frac{6\alpha^2}{10n^2}\right) + \left(1 -\frac{\alpha}{20} +  \frac{6\alpha^2}{4000} \right)+ \left(1 + \alpha  +  \frac{6\alpha^2}{10}\right)\notag \\
		&= n\left( 1 - \frac{\alpha(200-2406\alpha)}{4000n}  + \Oh\!\left(\frac{\alpha^2 }{n^2}\right)\right). \label{eq:2ndcondredux}
	\end{align}
	Thus we see by \eqref{eq:1stcondredux} and \eqref{eq:2ndcondredux} that if $\mathcal{G}^t$ holds and we take $\alpha<1/100$ and $n$ sufficiently large, then there exists some constant $c_2>0$ such that $A_i \leq n(1-c_2\alpha/n)$. Thus  \Cref{lem:filling} gives 
	\begin{align*}
		\ex{\Phi^{t + 1} \mid \mathfrak{F}^t, \mathcal{G}^t } &\leq \frac{1}{n}  \sum_{i=1}^n \Phi_i^{t} \cdot A_i + e^{3\alpha}\\ &\leq \frac{1}{n} \cdot n\left(1-\frac{c_2\alpha}{n}\right)\cdot  \sum_{i=1}^n \Phi_i^{t} + e^{3\alpha}\\ &\leq\left(1-\frac{c_2\alpha}{n}\right)\Phi^{t} + e^{3\alpha}, \end{align*} as claimed. \end{proof}

The next lemma shows that the event $\mathcal{G}^t$ given by \eqref{eq:eventGt} holds for sufficiently many rounds.

\begin{lem}\label{cor:enough_good_quantiles}
	Consider any allocation process satisfying conditions \POne and \WOne. For every integer $t_0 \geq 1$, there are at least $n/40$ rounds $t \in [t_0,t_0+n]$ with $(i)$ $\Delta^{t} \geq n/10$ or $(ii)$ $|B_{-}^{t}| \geq n/20$.
\end{lem}

\begin{proof} 
	We claim that if $\Delta^{s} \leq n/10$ for some round $s$, then for each round $t \in [s+n/5,s+9n/40]$ we have $|B_{-}^t| \geq n/20$ (deterministically). The lemma follows from this claim, since then we either have $\Delta^{t} \geq n/10$ for all $t \in [t_0,t_0+n/40]$, or, thanks to the claim there is a $s \in [t_0,t_0+n/40]$ such that for all $t \in [s+n/5,s+9n/40]$ (this interval has length $n/40$) we have $|B_{-}^t| \geq n/20$. 
	
	To establish the claim, assume we are at any round $s$ where $\Delta^{s} \leq n/10$. Then at most $n/2$ bins $i$ satisfy $|y_i^{s}| \geq 1/5$, and so at least $n/2$ bins satisfy $|y_i^{s}| < 1/5$; let us call this latter set of bins $\mathcal{B}:=\left\{ i \in [n]: |y_i^s| < 1/5 \right\}.$
	In the rounds $[s,s+9n/40]$, we can choose at most $9n/40$ bins in $\mathcal{B}$ that are overloaded (at the time when chosen), and then we place exactly one ball into them. Furthermore, in each round $t \in  [s,s+9n/40]$ we can take at most two bins in $\mathcal{B}$ which are underloaded at round $t$ and make them overloaded. Hence it follows that at least $n/2-2\cdot 9n/40 = n/20$ of the bins in $\mathcal{B}$ are not chosen in the interval $[s,s+9n/40]$. Consequently, these bins must be all underloaded in the interval $[s+n/5,s+9n/40]$.
\end{proof}

\subsection{Completing the Proof of Theorem~\ref{thm:filling_key}}\label{sec:complete_filling_proof}

We now introduce a new potential function $\tilde{\Phi}^t$ which is the product of $\Phi^t$ with two additional terms (and an additive centering term). These multiplying terms have been chosen based on the one round increments in the two statements in Lemmas \ref{lem:filling_good_quantile} and \ref{cor:enough_good_quantiles}. The purpose of this is that using Lemmas \ref{lem:filling_good_quantile} and \ref{cor:enough_good_quantiles} we can show that $\tilde{\Phi}^t$ is a super-martingale. We then use the super-martingale  property to bound the exponential potential at an arbitrary round.

Recall the event $\mathcal{G}^t := \left\{B_{-}^{t} \geq n/20 \right\}\cup \left\{ \Delta^t \geq n/10\right\}$ from \eqref{eq:eventGt}.
Now fix an arbitrary round $t_0 \geq 0$. Then, for any $s > t_0$, let $G_{t_0}^s$ be the number of rounds $r \in [t_0,s)$ satisfying $\mathcal{G}^r$, and let $B_{t_0}^s:=(s-t_0)-G_{t_0}^s$. Further, let the constants $c_1,c_2>0$ be as in \Cref{lem:filling_good_quantile}, let $c_3 := 2 e^{3 \alpha} \exp( c_2 \alpha)>0$, and then define the sequence $(\tilde{\Phi}^s)_{s \geq t_0} := (\tilde{\Phi}^s)_{s \geq t_0}(\alpha, t_0)$ with $\tilde{\Phi}^{t_0} := \Phi^{t_0}(\alpha)$, and for any $s > t_0 $,
\begin{equation}\label{eq:tildephi}
	\tilde{\Phi}^{s} := \Phi^{s}(\alpha) \cdot \exp\left( +\frac{c_2 \alpha}{n} \cdot G_{t_0}^{s-1} \right) \cdot \exp\left( -\frac{c_1 \alpha^2}{n} \cdot B_{t_0}^{s-1} \right) - c_3  \cdot (s-t_0).
\end{equation}
The next lemma proves that the sequence $(\tilde{\Phi}^{s})_{s \geq t_0}$, forms a super-martingale:

\begin{lem}
	\label{lem:filling_supermartingale}
	Let $0<\alpha<1/100$ be an arbitrary but fixed constant, and $t_0 \geq 0$ be an arbitrary integer. Then, for the potential $\tilde{\Phi}:=\tilde{\Phi}(\alpha, t_0)$ and any $s \in [t_0,t_0+n]$ we have	\[
	\ex{ \tilde{\Phi}^{s+1}  \mid \mathfrak{F}^s} \leq \tilde{\Phi}^{s}.
	\]
\end{lem}
\begin{proof}
	First, using the definition $\tilde{\Phi}^{s}$ from \eqref{eq:tildephi}, we rewrite  $\ex{\tilde{\Phi}^{s+1} \mid \mathfrak{F}^s}$ to give 
	\begin{align*}
		\lefteqn{ \ex{\tilde{\Phi}^{s+1} \mid \mathfrak{F}^s} } \\
		&=  \ex{ \Phi^{s+1}  \mid \mathfrak{F}^s}  \cdot \exp\left( \frac{c_2 \alpha}{n} \cdot G_{t_0}^{s} \right) \cdot \exp\left( - \frac{c_1 \alpha^2}{n} \cdot B_{t_0}^{s} \right)  - c_3 \cdot (s+1-t_0) \\ 
		& = \ex{ \Phi^{s+1} \mid \mathfrak{F}^s} \cdot \exp\left(\frac{\alpha}{n} \cdot ( c_2 \cdot \mathbf{1}_{\mathcal{G}^s} -c_1 \alpha\cdot(1 - \mathbf{1}_{\mathcal{G}^s}) ) \right) \\&\qquad  \cdot \exp\left( \frac{c_2 \alpha}{n} \cdot G_{t_0}^{s - 1} \right) \cdot \exp\left( - \frac{c_1 \alpha^2}{n} \cdot B_{t_0}^{s - 1} \right)  - c_3 - c_3 \cdot (s-t_0).
	\end{align*}

	We claim that it suffices to prove 
	\begin{align}\label{eq:sufficient}
		\ex{ \Phi^{s+1} \mid \mathfrak{F}^s} \cdot \exp\left(\frac{\alpha}{n} \cdot ( c_2 \cdot \mathbf{1}_{\mathcal{G}^s} -c_1 \alpha\cdot  (1 - \mathbf{1}_{\mathcal{G}^s}) \right) &\leq \Phi^{s} + c_3 \cdot \exp\left(- c_2 \alpha  \right).
	\end{align}

	Indeed, observe that $G_{t_0}^{s-1} \leq s-t_0 \leq n$, and so assuming \eqref{eq:sufficient} holds we have 
	\begin{align*}
		\ex{\tilde{\Phi}^{s+1} \mid \mathfrak{F}^s}     &\leq   \left( \Phi^{s} + c_3 \cdot \exp\left(- c_2 \alpha  \right) \right) 
		\cdot \exp\left( \frac{c_2 \alpha}{n} \cdot G_{t_0}^{s - 1} \right) \cdot \exp\left( - \frac{c_1 \alpha^2}{n} \cdot B_{t_0}^{s - 1} \right)\\ &\qquad  -c_3 - c_3 \cdot (s-t_0) \\
		&= \Phi^s \cdot \exp\left( \frac{c_2 \alpha}{n} \cdot G_{t_0}^{s - 1} \right) \cdot \exp\left( - \frac{c_1 \alpha^2}{n} \cdot B_{t_0}^{s - 1} \right) \\
		&\qquad+ c_3 \cdot \exp\left( - \frac{c_1 \alpha^2}{n} \cdot B_{t_0}^{s - 1} \right) -c_3 - c_3 \cdot (s-t_0) \\
		&\leq \tilde{\Phi}^s.
	\end{align*}
	To show \eqref{eq:sufficient}, we consider two cases based on whether $\mathcal{G}^s$ holds. 
	
	\smallskip 
	
	\noindent\textbf{Case 1} [$\mathcal{G}^s$ holds]. By \Cref{lem:filling_good_quantile}~$(ii)$ we have
	\[
	\ex{\Phi^{s+1} \mid \mathfrak{F}^s, \mathcal{G}^s} \leq \Phi^s \cdot \left(1 - \frac{c_2 \alpha}{n} \right) + e^{3\alpha} \leq \Phi^s \cdot \exp\left(- \frac{c_2 \alpha}{n} \right) +  e^{3\alpha} .
	\]
	Hence, if $\mathcal{G}^s$ holds then the left hand side of \eqref{eq:sufficient} is equal to  
	\begin{align*}
		\ex{\Phi^{s+1} \mid \mathfrak{F}^s, \mathcal{G}^s } \cdot \exp\left( \frac{\alpha}{n} \cdot c_2 \right) &\leq \left( \Phi^s \cdot \exp\left(- \frac{c_2 \alpha}{n} \right) + e^{3 \alpha} \right) \cdot \exp\left( \frac{\alpha}{n} \cdot c_2 \right) \\
		&\leq \Phi^s + 2 e^{3 \alpha} \\
		&= \Phi^s + c_3 \cdot \exp\left( -c_2 \alpha \right), 
	\end{align*}
	where the last line holds by definition of $c_3 = 2 e^{3 \alpha} \exp( c_2 \alpha)$.
	
	\smallskip 
	
	\noindent	\textbf{Case 2} [$\mathcal{G}^s$ does not hold]. \Cref{lem:filling_good_quantile}~$(i)$ gives the  unconditional bound 
	\[
	\ex{\Phi^{s+1} \mid \mathfrak{F}^s, \neg \mathcal{G}^s} \leq \Phi^s \cdot \left(1 + \frac{c_1 \alpha^2}{n} \right) + e^{3\alpha} \leq \Phi^s \cdot \exp\left(\frac{c_1 \alpha^2}{n} \right) + e^{3\alpha}.
	\]
	Thus, if $\mathcal{G}^s$ does not hold the left hand side of \eqref{eq:sufficient} is at most
	\[
	\ex{ \Phi^{s+1} \mid \mathfrak{F}^s, \neg \mathcal{G}^s } \cdot \exp\left( \frac{\alpha}{n} \cdot (-c_1 \alpha) \right) \leq \Phi^s + e^{3 \alpha} \leq \Phi^s + c_3 \cdot \exp\left( -c_2 \alpha \right),
	\] 
	which establishes \eqref{eq:sufficient} and the proof is complete.
\end{proof}

Combining \Cref{cor:enough_good_quantiles}, which shows that a constant fraction of rounds satisfy $\mathcal{G}^t$, with \Cref{lem:filling_supermartingale} establishes a multiplicative drop of $\Ex{ \Phi^{s} }$ (unless it is already linear), thus $\Ex{ \Phi^{m}}=\Oh(n)$. This is formalized in the proof below.

\begin{lem}\label{lem:exppot}
	There exists a constant $c_6>0$, such that for any allocation process satisfying conditions \POne and \WOne, and any round $m \geq 1$, we have $\Ex{\Phi^{m}} \leq c_6 \cdot n.$
\end{lem}
\begin{proof} 
	For any integer $t_0 \geq 1$, first consider rounds $[t_0,t_0+n]$. We will now fix the constant $\alpha := \min \{ 1/101, 1/20 \cdot c_2/c_1 \}$ in the exponential potential function $\Phi^t$ (thus this is also fixed in $\tilde{\Phi}^t)$. By \Cref{lem:filling_supermartingale}, $\tilde{\Phi}$ forms a super-martingale over $[t_0,t_0+n]$, and thus
	\begin{equation}\label{eq:supmart}
		\Ex { \tilde{\Phi}^{t_0+n} \; \Big| \; \mathfrak{F}^{t_0} } \leq \tilde{\Phi}^{t_0} = \Phi^{t_0}.
	\end{equation}
	Recalling the the definition \eqref{eq:tildephi} of $\tilde{\Phi}$ and $c_3 := 2 e^{3 \alpha} \exp( c_2 \alpha)>0$, we see that \eqref{eq:supmart} implies 
	\[
	\Ex{ \Phi^{t_0+n} \cdot \exp\left( +\frac{c_2 \alpha}{n} \cdot G_{t_0}^{t_0+n-1} \right) \cdot \exp\left( -\frac{c_1 \alpha^2}{n} \cdot B_{t_0}^{t_0+n-1} \right) - c_3 \cdot n
		\; \,\bigg| \,\; \mathfrak{F}^{t_0} } \leq \Phi^{t_0}.
	\]
	Rearranging this, and using that by~\Cref{cor:enough_good_quantiles}, $G_{t_0}^{t_0+n-1} \geq n/40$ holds deterministically, we obtain for any $t_0 \geq 1$
	\begin{align*}
		\Ex{ \Phi^{t_0+n} \, \mid \, \mathfrak{F}^{t_0} } &\leq \left( \Phi^{t_0} + c_3 \cdot n \right )
		\cdot \exp\left(  -c_2 \alpha \cdot \frac{1}{40} + c_1 \alpha^2 \cdot \frac{39}{40}\right), \\
		\intertext{and now using $\alpha = \min \{ 1/101, (1/40) \cdot c_2/c_1 \}$ and defining $c_4 := c_2/40^2$  yields} 
		\Ex{ \Phi^{t_0+n} \, \mid \, \mathfrak{F}^{t_0} } &\leq \left( \Phi^{t_0} + c_3 \cdot n \right ) \cdot \exp\left( -c_4 \alpha \right) \\
		&= \Phi^{t_0} \cdot \exp\left( -c_4 \alpha \right)
		+ c_3 \exp(-c_4 \alpha) \cdot n.
	\end{align*}
	It now follows by the second statement in \Cref{lem:geometric_arithmetic} with $a:=\exp(-c_4\alpha) < 1$ and $b:= c_3 \exp(-c_4 \alpha) \cdot n$ that for any integer $k \geq 1$,
	\begin{align} \label{eq:decrease_over_k_batches}
		\Ex{ \Phi^{n \cdot k}} &\leq 
		\Phi^{0} \cdot \exp\left( - c_4 \alpha \cdot k \right) + \frac{c_3 \exp(-c_4 \alpha) \cdot n}{1 - \exp(-c_4 \alpha)} \leq
		c_5 \cdot n,
	\end{align}
	for some constant $c_5 > 0$ as $\Phi^{0} \leq n$ holds deterministically.
	
	Hence for any number of rounds $t = k \cdot n + r$, where $k \geq 0$ and $1 \leq r < n$, we use~\Cref{lem:filling_good_quantile} (first statement) iteratively to conclude that
	\begin{align*}
		\Ex{ \Phi^{n \cdot k + r}  } & = \Ex{\Ex{ \Phi^{n \cdot k + r} \, \mid \, \mathfrak{F}^{n \cdot k } }}\notag \\
		& \leq \Ex{ \Phi^{n \cdot k}} \cdot \left(1 + \frac{c_1 \alpha^2}{n} \right)^{r} 
		+ n \cdot \left(1 + \frac{c_1 \alpha^2}{n} \right)^{n} \cdot e^{3\alpha} \notag \\
		&\leq c_5 \cdot n \cdot \exp( c_1 \alpha^2) + n \cdot \exp( c_1 \alpha^2) \cdot e^{3 \alpha}\\ &\leq c_6 \cdot n, 
	\end{align*}
	for some constant $c_6:=c_6(c_1, c_2) > 0$, as claimed.
\end{proof}

Having established $\Ex{ \Phi^{m}}=\Oh(n)$ in the previous lemma, proving the gap bound is simple

\fillingkey*

\begin{proof} It follows directly from Lemma \ref{lem:exppot} and Markov's inequality that for any $m \geq 1$,
	\[
	\Pro{ \Phi^{m} \leq c_6 \cdot n^{3} } \geq 1-n^{-2}.
	\]
	Since $\Phi^{m} \leq c_6 \cdot n^{3}$ implies $\Gap(m)=\Oh(\log n)$, the proof is complete.
\end{proof}
In addition to a gap bound that holds w.h.p.~in $n$ (the number of bins) we also obtain the following bound which holds w.h.p.~in $m$ (the number of balls). A similar guarantee with an exceptional probability given as a function of the number of balls was recently given by~\cite{BK22}.

\setcounter{thm}{\value{lem}}
\begin{thm}
	\label{thm:highprobinm}
	There exists a constant $C > 0$ such that for any allocation process satisfying conditions \POne and \WOne, and any round $m \geq n$, we have 
	\[
	\Pro{\Gap(m) \leq C \log m} \geq 1 - m^{-2},
	\]
	and further, this implies that
	\[
	\Pro{\bigcap_{m \in [n, \infty)} \left\{\Gap(m) \leq C \cdot \log m \right\}} \geq 1 - \frac{1}{2} n^{-1}.
	\]   
\end{thm}
\begin{proof}
	It follows directly from \Cref{lem:exppot} and Markov's inequality that for any $m \geq n$,
	\[
	\Pro{ \Phi^{m} \leq c_6 \cdot n \cdot m^{2} } \geq 1-m^{-2}.
	\]
	Since $\Phi^{m} \leq c_6 \cdot n \cdot m^{2} \leq c_6 \cdot m^3$ implies $\Gap(m)=\Oh(\log m)$, the first statement follows.
	
	For the second statement, by taking the union bound over all steps $n, n+1, \ldots$, we have that
	\[
	\Pro{\bigcap_{m \in [n, \infty)} \left\{ \Gap(m) \leq C \cdot \log m \right\}} \geq 1 - \sum_{m \in [n, \infty)} t^{-2} \geq 1 - \int_{n-1}^{\infty} t^{-2} \,\mathrm{dt} \geq 1 - \frac{1}{2}n^{-1},
	\]
	using that the function $f(t) = t^{-2}$ is convex for $t > 0$.
\end{proof}

\subsection{Proof of Proposition~\ref{Pro:negbias}}\label{sec:biasneg}

The \Packing process with a non-uniform probability vector is not a \Filling process and so our general $\Oh(\log n)$ gap bound does not apply. Thus we now give a short and basic proof that the bound does not diverge with $m$.

\newnegbias*

\begin{proof}
	We will analyze the change of the absolute value potential over an arbitrary round $t$. We consider two cases based on the load of the sampled bin $i \in [n]$:

	\textbf{Case 1 [$y_i^t \geq 0$]:} When an overloaded bin is allocated to, then it contributes $+1$ to the absolute value potential and each underloaded bin contributes $1/n$. So, since there are at most $n-1$ underloaded bins
	\[
	\Delta^{t+1} - \Delta^t \leq 1 + \frac{1}{n} \cdot (n-1) \leq 2.
	\]
	
	\textbf{Case 2 [$y_i^t < 0$]:} When an underloaded bin is allocated to, then the sampled bin $i$ contributes at most $-\lceil - y_i^t \rceil + 3$ to the change in potential. In addition, the average changes by $(\lceil - y_i^t \rceil + 1)/n$. As a result (of the average change alone), each overloaded bin contributes $-(\lceil - y_i^t \rceil + 1)/n$ and each underloaded bin contributes $(\lceil - y_i^t \rceil + 1)/n$. Since there is always at least one overloaded bin (that will also remain overloaded after the change of the average), the aggregate change when allocating to an underloaded bin can be bounded as follows,
	\[
	\Delta^{t+1} - \Delta^t \leq -\lceil - y_i^t \rceil + 3 -\frac{\lceil - y_i^t \rceil + 1}{n} + (n-1) \cdot \frac{\lceil - y_i^t \rceil + 1}{n} = 2\cdot \frac{y_i^t}{n} + 4 -\frac{2}{n}  \leq   \frac{2y_i^t}{n} + 4.
	\]
	Hence, by combining these two cases, the expected change of the absolute value potential is given by
	\[
	\Ex{\Delta^{t+1} - \Delta^t \mid x^t} \leq \sum_{i : y_i^t \geq 0} 2 \cdot p_i^t + \sum_{i : y_i^t < 0} \Big(\frac{2y_i^t}{n} + 4\Big) \cdot p_i^t \leq - \frac{\Delta^t}{an^3} + 6,
	\]
	using that $p_i^t \geq \frac{1}{an}$ and for the bin $i$ with the minimum load at round $t$, $y_i^t \leq -\frac{\Delta^t}{2n}$. Using induction, we show that $\Ex{\Delta^t} \leq 6an^3$. As $\Delta^0 = 0$ and assuming $\Ex{\Delta^t} \leq 6an^3$, we have 
	\[
	\Ex{\Delta^{t+1}} = \Ex{\Ex{\Delta^{t+1} \mid \Delta^t}} \leq \Ex{\Delta^t}  \cdot \Big(1 - \frac{1}{an^3}\Big) + 6 \leq 6an^3 - \frac{6an^3}{an^3} + 6  = 6an^3.
	\]The final part of the result now follows from an application of Markov's inequality and using that $\Gap(m) \leq \Delta^m$.
\end{proof}

\section{Unfolding General Filling Processes}\label{sec:unfolding_general}
In this section, we prove our results on (general) unfoldings of \Filling processes, though first we will recall the definition of the unfolding of a process from Page \pageref{folding}:

\unfolding

We now show that our notion of ``unfolding'' can be applied to capture the \Memory process.

\memorycompressed*

\begin{proof}
	As in the definition of unfolding, we denote the load vector of $P$ after round $t$ by $x^t$, and the load vector of a suitable unfolding $U(P)$ after the $s$-th atomic allocation by $\widehat{x}^{\,s}$. We also denote the corresponding normalized load vectors by $y^t$ and $\widehat{y}^{\,s}$ respectively. 
	
	We will construct by induction, a coupling between a suitable allocation process $P$, satisfying \POne and \WOne, and an unfolding $U(P)$ which follows the distribution of \Memory. That is for every round $t \geq 0$ of $P$, there exists a (unique) atomic allocation $s=a(t) \geq t$ in $U(P)$, such that $x^t=\widehat{x}^{\,a(t)}$, and $U(P)$ is an instance of \Memory. 
	
	Assume that for a suitable unfolding of the process $P$, the load configuration of $P$ after round $t$ equals the load configuration of $U(P)$ after atomic allocation $s=a(t)$, i.e., $x^t = \widehat{x}^{\,a(t)}$. In case the cache is empty (which happens only at the first round $s=0$), \Memory will sample a uniform bin $i$ (which satisfies \POne). If the cache is not empty, \Memory will take as bin $i$ the least loaded of the bin in the cache and a uniformly chosen bin. This produces a distribution vector that is majorized by \OneChoice (thus satisfies \POne, again). Thus we may couple the two instances such that process $P$ samples the same bin $i$ in round $t$ and atomic allocation $a(t)$, respectively. We continue with a case distinction concerning the load of bin $i$ at round $t$:
	
	\textbf{Case 1} [The bin $i$ is overloaded, i.e., $y_i^{t} = \widehat{y}_i^{\,a(t)} \geq 0$]. Then \Memory and $P$ both place one ball into bin $i$, satisfying \WOne. Further, $P$ proceeds to the next round and $U(P)$ proceeds to the next atomic allocation, which means that the coupling is extended.
	\smallskip
	
	\textbf{Case 2} [The bin $i$ is underloaded, i.e., $y_i^{t} = \widehat{y}_i^{\,a(t)} < 0$]. Then we can deduce by definition of \Memory that it will place the next $\big\lceil -y_i^{a(t)} \big\rceil +1$ balls in some way that it deterministically satisfies the following conditions:
	$(i)$ the first $\big\lceil -y_i^{a(t)} \big\rceil $ balls are placed into bins which have a normalized load $<0$ at the atomic allocation $a(t)$, $(ii)$ one ball is placed into a bin with normalized load $<1$ at the atomic allocation $a(t)$. This follows since bin $i$ gets stored in the cache and at each atomic allocation $ j= W^{t-1} + 1, \dots , W^t $ the process has access to a cached bin with normalized load at most $y_i^{a(t)} +j-1$. This satisfies \WOne so we can continue the coupling. 
	
	We have thus constructed a process $P$, such that some unfolding $U=U(P)$ of $P$ is an instance of  \Memory. \end{proof}

We now restate and prove the general gap bound for the unfolding of the processes. 

\sloweddown*

\begin{proof}

	We will re-use the constants $\alpha \in (0,1)$ and $c_6 > 0$ given by Lemma \ref{lem:exppot}. To begin, define
	\[
	\mathcal{B}:= \left| \left\{ t \in [1,m] \colon \Phi^t \geq c_6 \cdot n^{6+c}  \right\}  \right|,
	\] 
	which is the number of ``bad'' rounds of the \Filling process $P$.  Let $w^{t}:=W^{t+1}-W^{t}$ denote the  number of balls allocated in round $t$. We continue with a case distinction for each round $t$ whether $t \in \mathcal{B}$ holds:
	
	\begin{itemize}
		\item \textbf{Case 1} [$t \not\in \mathcal{B}$].
		By definition, for a round $t \not\in B$ we have $\Phi^{t} < c_6 n^{6+c}$.  Further, $\Phi^t < c_6 n^{6+c}$ implies $\Gap_{P}(t)< 2\cdot \frac{6+c}{\alpha} \cdot \log n$, for sufficiently large $n$. Now let $a(t),a(t)+1,\ldots,a(t)+w^t-1$, $w^t:=\lceil -y_i^t \rceil+1$ be the atomic allocations in $U(P)$ corresponding to round $t$ in $P$. Since all allocations of $U(P)$ are to bins with normalized load at most $1$ before the allocation, we conclude $\Gap_{U(P)}(s) \leq \max\{\Gap_{P}(t),2\} < 2\cdot \frac{6+c}{\alpha} \cdot \log n$ for all $s \in [a(t),a(t)+w^t-1]$.
		
		\item \textbf{Case 2} [$t \in \mathcal{B}$]. 
		We will use that for any $0 < \alpha < 1$ and $t \geq 0$ we have 
		\[
		\Delta^t \leq   2n \cdot \left( \frac{1}{\alpha}\log \Phi^t + 1\right).
		\]
		To see this, observe that $\sum_{i\in B_+^t}y_i^t = -
		\sum_{i\in B_-^t}y_i^t$ and thus \begin{equation}\label{eq:delbdd}\Delta^t\leq 2\sum_{i\in B_+^t}y_i^t\leq 2\sum_{i\in B_+^t}y_i^t\mathbf{1}_{y_i^t\geq 2} + 2n.\end{equation} Now, note that since $ \Phi^t =  \sum_{i\in[n] : y_{i}^{t}\geq 2} \exp\left( \alpha \cdot  y_i^t \right)$, if $ \Phi^t  \leq \lambda$ then  $y_i^t \leq \frac{1}{\alpha} \cdot \log \lambda$ for all $i\in[n]$ with $y_{i}^t\geq 2$. Thus by \eqref{eq:delbdd} we have $\Delta^t\leq 2n\cdot (1/\alpha)\log\Phi^t + 2n$ as claimed.
		
		So for a round $t \in \mathcal{B}$, we have 
		\begin{align}
			w^{t} \leq \Delta^{t} + 1 \leq \frac{2n}{\alpha}\log \Phi^t + 2n + 1 \leq c_\alpha n\log \Phi^t, \label{eq:weight_bound}
		\end{align}
		which holds deterministically for some constant $c_\alpha > 0$. Again, every such round $t \in \mathcal{B}$ of $P$ corresponds to the atomic allocations $a(t),a(t)+1,\ldots,a(t)+w^{t}-1$ in $U(P)$, and we will (pessimistically) assume that the gap in all those rounds is large, i.e., at least $2\cdot \frac{6+c}{\alpha} \cdot \log n$.
	\end{itemize}
	
	Let $C:= 2\cdot \frac{6+c}{\alpha}$. Then, by the above case distinction and \eqref{eq:weight_bound}, we can upper bound the number of rounds $s$ in $[1,m]$ of $U(P)$ where $\Gap_{U(P)}(s) < C\cdot \log n$ does not hold as follows:
	\begin{align}
		\left|\left\{ s \in [1,m] \colon \Gap_{U(P)}(s) \geq C\cdot \log n \right\}\right| \leq \sum_{t=1}^m \mathbf{1}_{t \in \mathcal{B}} \cdot w^{t} \leq  c_{\alpha} n  \cdot \sum_{t=1}^m \mathbf{1}_{t \in \mathcal{B}} \cdot \log \Phi^t  . \label{eq:bad_times}
	\end{align}
	Next define the sum of the  exponential potential function over rounds $1$ to $m$ as
	\[
	\Phi := \sum_{t=1}^{m} \Phi^t.
	\]
	
	Then by Lemma \ref{lem:exppot}, there is a constant $c_6 > 0$ such that $\Ex{ \Phi^t } \leq c_6 \cdot n$, and hence 
	\[
	\Ex{ \Phi } = \sum_{t=1}^m \Ex{ \Phi^t } \leq m \cdot c_6 \cdot n.
	\]
	By Markov's inequality,
	\begin{equation}
		\Pro{  \Phi \leq m \cdot c_6 \cdot n^{3} } \geq 1 - n^{-2}.\label{eq:markov_inequality}
	\end{equation}
	Note that conditional on the above event occurring, the following bound holds deterministically:
	\begin{equation}\label{eq:bbound}
		| \mathcal{B} | \leq \frac{ m c_6 n^{3} }{ c_6 n^{6+c} } \leq m \cdot n^{-3-c}.   
	\end{equation} 
	Also observe that we have  
	\begin{equation}\label{eq:logsumbd}\sum_{t=1}^m \mathbf{1}_{t \in \mathcal{B}} \cdot  \log \Phi^t =   \sum_{t \in \mathcal{B}} \log \Phi^t \leq   \sum_{t \in \mathcal{B}}   \log \Phi \leq    |\mathcal{B}|  \cdot  \log \Phi.  \end{equation}
	Thus, if the event $\Phi \leq m \cdot c_6 \cdot n^3$ occurs, then by \eqref{eq:bad_times}, \eqref{eq:bbound} and \eqref{eq:logsumbd} we have  
	\begin{align*}
		\left|\left\{ s \in [1,m] \colon \Gap_{U(P)}(s) \geq C \cdot \log n \right\}\right| &\leq c_{\alpha}n \cdot  |\mathcal{B}| \cdot   \log \Phi \\   
		&\leq c_{\alpha}n \cdot  (m \cdot n^{-3-c}) \cdot   \log\left(  m \cdot c_6 \cdot n^3\right)\\
		&\leq c_{\alpha}' m \cdot n^{-2-c}  \cdot \left( \log  m  + \log n \right)\\
		&\leq  m \cdot n^{-c}  \cdot \log  m ,
	\end{align*}
	where the third inequality is for some constant $c'_\alpha$ depending on $\alpha$ and the last inequality holds since  $\alpha>0$ is a small but fixed constant thus $c_\alpha'$ is constant. Since the inequality holds for \emph{any} unfolding of $P$, once the event in \eqref{eq:markov_inequality} occurs, we obtain the result.
\end{proof}

\section{Lower Bounds on the Gap}\label{sec:lower}

In this section we shall prove several lower bounds for \Filling processes. In \Cref{sec:uniform_lb}, we prove an $\Omega(\frac{\log n}{\log \log n})$ lower bound for the \Packing and \TightPacking processes for $m=\Oh(n)$ rounds. In \Cref{sec:packing_tight_lb}, for the \Packing process, we prove a tight $\Omega(\log n)$ lower bound for any $m = \Omega(n \log n)$ (see \Cref{tab:overview} for a concise overview of our lower and upper bounds). 
\subsection{Lower Bound for Uniform Processes} \label{sec:uniform_lb}

Since \Packing and \TightPacking use a uniform probability vector, the result below immediately yields a gap bound of $\Omega(\frac{\log n}{\log \log n})$ for these processes for $m=\Oh(n)$ rounds. 

\begin{lem} \label{lem:uniform_prob_lower_bound}
	For any allocation process with a uniform probability vector that satisfies condition \WOne we have \[
	\Pro{\Gap\left(\frac{n}{2}\right) \geq \frac{1}{2} \cdot \frac{\log n}{\log \log n}} \geq 1 - o(1).
	\]
\end{lem}
\begin{proof}
	We begin with the simple observation that for any round $t \leq n/2$, we have $W^t \leq 2t - 1$ (so in particular, $W^{n/2} \leq n$, i.e., at round $n/2$ we have at most $n$ balls allocated in total.) For $t = 1$, the statement is true, since the sampled bin is overloaded and so we allocate exactly one ball. Assuming $W^{t} \leq 2t -1$ holds for some $t \leq n/2 -1$, then at the beginning of round $t + 1$, the average load is at most $(2t - 1)/n < 1$. Hence even if we sample an empty bin $i \in [n]$, then $\lceil - y_i^t \rceil \leq 1$, so we can allocate at most $\lceil - y_i^t \rceil + 1 \leq 2$ balls. Hence, $W^{t+1} \leq W^t + 2 \leq 2(t+1) - 1$, which completes the induction.

	Next note that in the first $n/4$ rounds, whenever we sample an empty bin $i^t$, we turn at least one underloaded bin (which may be different from $i^t$) into an overloaded bin, and the bin remains overloaded at least until round $n/2$. Using a standard concentration inequality (e.g., Method of Bounded Differences), it follows that with probability $1-n^{-\omega(1)}$, we create at least $n/16$ overloaded bins during the first $n/4$ rounds. In the following, let us denote such a set of bins $\mathcal{B}$, and w.l.o.g.~assume $|\mathcal{B}|=n/16$.
	
	Consider now the rounds $n/4+1,n/4+2,\ldots,n/2$. Whenever a bin from $\mathcal{B}$ is sampled, its load is incremented by $1$. Further, with probability $|\mathcal{B}|/n = 1/16$, a bin from $\mathcal{B}$ is sampled. Using a Chernoff bound, with probability $1-n^{-\omega(1)}$, we sample a bin from $\mathcal{B}$ in the rounds $n/4+1,\ldots,n/2$ at least $n/128$ times. Thus with probability $1-2 n^{-\omega(1)}$, we can couple the allocations of at least $n/128$ balls with a \OneChoice process with $n/128$ balls into $n/16$ bins.

	By e.g.,~\cite[Theorem 1]{RS98}, in the \OneChoice process with $n/128$ balls into $n/16$ bins, with probability at least $1-o(1)$ there is a bin $i \in [n]$ which will receive at least $  \frac{2}{3}\cdot \frac{\log n}{\log \log n}$ balls. Hence, at the end of round $n/2$, with probability at least $1 - o(1)$, we have that\[
	\Gap(n/2) \geq \frac{2}{3}\cdot \frac{\log n}{\log \log n} - \frac{W^{n/2}}{n} \geq \frac{2}{3}\cdot   \frac{\log n}{\log \log n} - 1 \geq  \frac{1}{2}\cdot  \frac{\log n}{\log \log n},  
	\]as claimed.
\end{proof}

The next lower bound applies to a class of allocation processes, where $(i)$ balls can only be allocated to the (uniformly) sampled bin, but $(ii)$ the number of allocated balls is allowed to depend on $\mathfrak{F}^t$ (as well as the bin sample).

\begin{lem} \label{lem:packing_lower_bound}
	Consider any allocation process, which at each round $t \geq 0$, picks a bin $i^t$ uniformly at random. Furthermore, assume that at any round $t \geq 0$ the allocation process increments the load of bin $i^t$ by some function $f^t \geq 1$, which may depend on $\mathfrak{F}^{t}$ and $i^t$. Then, 
	\[
	\Pro{ \Gap\left(\frac{n}{2}\right) \geq \frac{1}{10}\cdot \frac{\log n}{\log \log n} } \geq 1-o(1).
	\]
\end{lem}
\begin{proof}

	Recall the fact (see, e.g.,~\cite[Theorem 1]{RS98}) that in a \OneChoice process with $n/2$ balls into $n$ bins, with probability at least $1-o(1)$ there is a bin $i \in [n]$ which will be chosen at least $\kappa \cdot  \frac{\log n}{\log \log n}$ times during the first $n/2$ allocations, where we can take $\kappa=1/5$. Hence in our process
	\[
	x_i^{n/2} \geq  \left(\kappa \cdot \frac{\log n}{\log \log n}\right) \cdot 1,
	\]for some $i\in[n]$ w.h.p.\ as at least one ball is allocated at each round. 
	
	Let us define $f^{*}:=\max_{1 \leq t \leq n/2} f^t$ to be the largest number of balls allocated in one round. Then, clearly, $W^{n/2} \leq (n/2) \cdot f^{*}$. Furthermore, there must be at least one bin $j \in [n]$ which receives $f^*$ balls in one of the first $n/2$ rounds. Thus for any such bin, 
	\[
	x_j^{n/2} \geq f^*,
	\]
	and therefore the gap is lower bounded by
	\begin{align*}
		\Gap(n/2) &\geq \max \left\{x_i^{n/2}, x_j^{n/2} \right\} - \frac{W^{n/2}}{n} \geq \max \left\{ \kappa \cdot \frac{\log n}{\log \log n}, f^{*} \right\} - \frac{f^{*}}{2}.
	\end{align*}
	If $f^{*} \geq \kappa \cdot \frac{\log n}{\log \log n}$, then the lower bound is $\frac{1}{2} f^{*} \geq \frac{\kappa}{2} \cdot \frac{\log n}{\log \log n}$. Otherwise, $f^{*} < \kappa \cdot \frac{\log n}{\log \log n}$, and the lower bound is at least
	$ \kappa \cdot \frac{\log n}{\log \log n} - \frac{\kappa}{2} \cdot \frac{\log n}{\log \log n} = \frac{ \kappa}{2} \cdot \frac{\log n}{\log \log n}$, where we recall $\kappa=1/5$.
\end{proof}

\subsection{An Improved Lower Bound for Packing} \label{sec:packing_tight_lb}

In this section we prove a lower bound that is tight up to a multiplicative constant for the \Packing process when  $m = \Omega(n \log n)$. The result is proven using the following two technical lemmas that will also be useful when proving the upper bound in Theorem \ref{thm:sample_efficient}, the result on the throughput of a process.   The first lemma concerns the absolute value potential $\Delta^t$. 
\begin{lem}\label{eq:exp_delta_bound}	 There exists a constant $\tilde{c}>0$, such that for any allocation process satisfying conditions \POne and \WOne, and any round $t \geq 1$, we have  $\Ex{\Delta^t} \leq  \tilde{c} n$. 
\end{lem}
\begin{proof}
	We will be using the exponential potential $\Phi := \Phi(\alpha)$ with  $\alpha := \min \{ 1/101, 1/20 \cdot c_2/c_1 \}$ where $c_1, c_2 > 0$ are the constants defined in \Cref{lem:filling_good_quantile}. Recall that\[
	\Phi^t = \sum_{i : y_i^t \geq 2} e^{\alpha y_i^t}. 
	\]Hence, using that $e^z \geq 1 + z$ for any $z$,
	\begin{align} \label{eq:phi_geq_delta}
		\Phi^t \geq \sum_{i : y_i^t \geq 2} (1 + \alpha y_i^t) \geq \sum_{i : y_i^t \geq 2} \alpha y_i^t \geq \sum_{i : y_i^t \geq 0} \alpha (y_i^t - 2) \geq \alpha \cdot \Big(\frac{\Delta^t}{2} - 2n \Big).
	\end{align}
	By Lemma \ref{lem:exppot} there exists some constant $c_6 > 0$ such that for any $t \geq 0$, \[
	\Ex{\Phi^t} \leq c_6 \cdot n.
	\]
	
	Rearranging  \eqref{eq:phi_geq_delta} and using this bound on $\Ex{\Phi^t}$ implies that
	\[
	\Ex{\Delta^t} \leq  \frac{2}{\alpha} \cdot \Ex{\Phi^{t}} + 4n \leq \frac{2c_6}{\alpha} \cdot n + 4 n =: \tilde{c} n,\]
	with the constant $\tilde{c}:=\frac{2c_6}{\alpha}+4$, as claimed.
\end{proof}

For the second lemma, recall that $W^t$ is the number of balls allocated up to round $t$. 
\begin{lem}\label{lem:supermartin}
	Let $\tilde{c} > 0$ be the constant from Lemma \ref{eq:exp_delta_bound}. Then, for any allocation process satisfying conditions \POne and \WOne, and any round $m\geq t_0 \geq 0$, we have \[\ex{W^m-W^{t_0}} \leq (\tilde{c}+2)\cdot (m-t_0).\] 
\end{lem}

\begin{proof} For any $t_0\geq 0$, we fix $\tilde{W}^{t_0} := 0$, and for any $t > t_0$ we let
	\begin{equation}\label{eq:supermartin}
		\tilde{W}^t := W^t - W^{t_0} - 2 \cdot (t - t_0) - \frac{1}{n} \cdot \sum_{s = t_0 + 1}^{t} \Delta^s.
	\end{equation}
	We shall show that $\tilde{W}^t$ is a supermartingale. Taking expectations over one round for $t > 0$ gives,
	\begin{align*}
		\Ex{\tilde{W}^{t} \,\, \left\vert \,\,  \mathfrak{F}^{t-1} \right.}   &= W^{t-1} - W^{t_0} + \sum_{i : y_i^{t-1} \geq 0} p_i + \sum_{i : y_i^{t-1} < 0} p_i \cdot (1 + \left\lceil -y_i^{t-1} \right\rceil)\\ &\qquad -  2 \cdot (t - t_0) - \frac{1}{n} \cdot \sum_{s = t_0+1}^{t} \Delta^s. \end{align*}
	Now, since $p_i=1/n$ and \[\sum_{i \colon y_i^{t-1} <0} \lceil -y_i^{t-1} \rceil \leq \sum_{i \colon y_i^{t-1} <0} (-y_i^{t-1}+1) \leq \Delta^{t-1}/2 + |B_{-}^{t-1}|,\] we have 
	\begin{align*}\Ex{\tilde{W}^{t} \,\, \left\vert \,\,  \mathfrak{F}^{t-1} \right.}  	&  \leq  W^{t-1} - W^{t_0} + |B_{+}^{t-1}| \cdot \frac{1}{n}  + 2 \cdot |B_{-}^{t-1}| \cdot \frac{1}{n} + \frac{\Delta^{t-1}/2}{n} \\ &\qquad - 2 \cdot (t - t_0) - \frac{1}{n} \cdot \sum_{s = t_0+1}^{t} \Delta^s \\
		&  \leq  W^{t-1} - W^{t_0} + 2 + \frac{\Delta^{t-1}}{n} - 2 \cdot (t - t_0) - \frac{1}{n} \cdot \sum_{s = t_0+1}^{t} \Delta^s \\
		&  = W^{t-1} - W^{t_0} - 2 \cdot (t-t_0-1) - \frac{1}{n} \cdot \sum_{s = t_0+1}^{t-1} \Delta^s \\
		&  = \tilde{W}^{t-1},
	\end{align*} 
	proving $\tilde{W}^t$ is a supermartingale. Thus, since $\tilde{W}^{t_0} = 0$, for any $m\geq t_0$ we have
	\[
	\Ex{\tilde{W}^{m}} \leq \tilde{W}^{t_0} = 0.
	\]
	Finally, recalling that $\Ex{\Delta^t} \leq  \tilde{c} n$ for any $t \geq 0$ by Lemma \ref{eq:exp_delta_bound}, we have
	\[
	\Ex{W^{m} - W^{t_0}} \leq \frac{1}{n} \cdot \Ex{\sum_{s = t_0 + 1}^{m} \Delta^s} + 2 \cdot (m - t_0) \leq (\tilde{c} + 2) \cdot (m - t_0),  
	\]by \eqref{lem:supermartin} for the constant $\tilde{c}>0$ given by Lemma \ref{eq:exp_delta_bound}.\end{proof}

We are now ready to restate and prove the main result in this section. 

\packlower* 

\begin{proof} Recall that  for any $m\geq t_0\geq 0$ we have $\Ex{W^{m} - W^{t_0}} \leq(\tilde{c} + 2) \cdot (m - t_0) $ by Lemma \ref{lem:supermartin}. Since $W^{m} \geq W^{t_0}$ holds deterministically, we can apply Markov's inequality to give
	\begin{align} \label{eq:small_w_t}
		\Pro{W^{m} \leq W^{t_0} + 4 \cdot (\tilde{c} + 2) \cdot (m - t_0)} \geq \frac{3}{4}.
	\end{align}
	
	Recall that $\Ex{\Delta^{t_0}} \leq  \tilde{c} n$ for any $t_0 \geq 0$ by Lemma \ref{eq:exp_delta_bound}. Then, using Markov's inequality
	\begin{align} \label{eq:delta_t0_small}
		\Pro{\Delta^{t_0} \leq 8\tilde{c}n} \geq \frac{7}{8}.
	\end{align}Fix $t_0 := m - \kappa n \log n\geq 0$, for some constant $\kappa > 0$ to be defined shortly, and define the set 
	\[
	\mathcal{B} := \left\{ i \in [n] : y_i^{t_0} \geq - 128\cdot\tilde{c}\right\}.
	\]
	Note that any $i \in \mathcal{B}$ satisfies $x_i^{t_0} \geq \frac{W^{t_0}}{n} - 128 \tilde{c}$.
	Further, if  $\{\Delta^{t_0} \leq 8\cdot\tilde{c}n\}$ holds, then $|\mathcal{B}| \geq (15/16) n$ holds deterministically. By a coupling with the \OneChoice process with $\tilde{m} := \kappa n \log n$ balls and $n$ bins, with probability at least $1 - o(1)$ there is a bin that is sampled at least $\frac{\sqrt{\kappa}}{10} \cdot \log n$ times (e.g.,~\cite{RS98}). By symmetry, with probability $15/16$, this maximally chosen bin is in $\mathcal{B}$. Since each time a bin is sampled by \Packing, it gets at least one ball, thus
	\begin{align} \label{eq:one_choice_argument}
		\Pro{\bigcup_{i \in [n]} \left\lbrace x_i^{m} \geq \frac{W^{t_0}}{n} + \frac{\sqrt{\kappa}}{10} \cdot \log n -128 \tilde{c} \right\rbrace \cap \left\lbrace \Delta^{t_0} \leq 8\tilde{c}n \right\rbrace } &\geq 1 - o(1) - \frac{1}{8} - \frac{1}{16}    \geq \frac{3}{4}.
	\end{align}
	Assuming that the events \[\bigcup_{i \in [n]} \left\lbrace x_i^{m} \geq \frac{W^{t_0}}{n} +  \frac{\sqrt{\kappa}}{10} \cdot \log n - 128\tilde{c} \right\rbrace,\] and \[\left\lbrace W^{m} \leq W^{t_0} + 4 \cdot (\tilde{c} + 2) \cdot (m - t_0) \right\rbrace,\] hold, then by choosing $\kappa := (\frac{1}{100 \cdot (\tilde{c} + 2)})^2$, there exists some $i\in [n]$ such that
	\[
	y_i^{m} \geq \frac{W^{t_0}}{n} + \frac{\sqrt{\kappa}}{10} \cdot \log n -128\tilde{c} - \frac{W^{m}}{n} \geq \frac{\sqrt{\kappa}}{10} \cdot \log n - \kappa \cdot 4 \cdot (\tilde{c} + 2) \cdot \log n -128 \tilde{c} \geq \frac{\sqrt{\kappa}}{20} \cdot \log n.
	\]
	Taking the union bound over \eqref{eq:small_w_t} and \eqref{eq:one_choice_argument},
	\[
	\Pro{\Gap(m) \geq \frac{\sqrt{\kappa}}{20} \cdot \log n} \geq \frac{3}{4} - \frac{1}{4} = \frac{1}{2},\]as claimed.
\end{proof}

\section{Throughput and Sample-Efficiency of Filling Processes}\label{sec:sampleeff}

Recall that $W^t := \sum_{i\in[n]} x_i^t$ is the total number of balls allocated by round $t$. In this section, we will bound the expectation of the throughput $\mu^t$ of a \Filling process at an arbitrary round $t \geq 0$, which was defined as\[
\mu^t = \frac{W^t}{t}.
\]

\sampleefficient*

Note that for processes which make one sample per round, i.e., $S^t = t$, such as the \Packing process, throughput coincides with the sample efficiency $\eta^t$ of the process, i.e., $\mu^t = \eta^t$, where
\[
\eta^t = \frac{W^t}{S^t}.
\]
Therefore, \Cref{thm:sample_efficient} implies that \Packing is more sample-efficient than \OneChoice by a constant factor in expectation. 

\PackingSampleEfficient*

The empirical results of \Cref{fig:packing_sample_efficiency} strongly support this, suggesting that the sample efficiency of \Packing is around $3/2$ on average.

\begin{proof}[Proof of \Cref{thm:sample_efficient}]We start with the following general expression for the expected number of balls allocated in an arbitrary round $t$, 
	\begin{equation}\label{eq:prediff} \Ex{W^{t+1}-W^{t} \;\big| \;\mathfrak{F}^t } = \sum_{i: y_i^t\geq 0} 1 \cdot p_i + \sum_{i: y_i^t<0} \left(1 + \left\lceil -y_i^t\right\rceil  \right)\cdot p_i  = 1 + \sum_{i: y_i^t<0} \left\lceil -y_i^t\right\rceil \cdot p_i.  \end{equation}

	Let $i^*$ be the smallest index such that $y_{i^*}^t<0 $ and observe that the sequence $\left\lceil -y_n^t\right\rceil, \dots,  \left\lceil -y_{i^*}^t\right\rceil$ is non-negative and non-increasing. Additionally, $\sum_{i=1}^k 1/n \geq \sum_{i=1}^k p_i$ for all $k\in [n]$ by condition \POne, and thus $\sum_{i=k}^n 1/n \leq \sum_{i=k}^n p_i$ for any $i^*\leq k\leq n$. Hence, by \Cref{lem:quasilem}, we have\[\sum_{i: y_i^t<0} \left\lceil -y_i^t\right\rceil \cdot p_i = \sum_{i=i^*}^n \left\lceil -y_i^t\right\rceil \cdot p_i \geq \sum_{i=i^*}^n \left\lceil -y_i^t\right\rceil \cdot \frac{1}{n} = \sum_{i: y_i^t<0} \left\lceil -y_i^t\right\rceil \cdot \frac{1}{n}. \]
	Therefore, it follows from \eqref{eq:prediff} that 
	\begin{equation}\label{eq:1diff} \Ex{W^{t+1}-W^{t} \;\big| \;\mathfrak{F}^t } = 1 + \sum_{i: y_i^t<0} \left\lceil -y_i^t\right\rceil \cdot p_i\geq  1 + \sum_{i: y_i^t<0} \left\lceil -y_i^t\right\rceil \cdot \frac{1}{n}.  \end{equation} Since $\Delta^t=\sum_{i=1}^n |y_i^t|  $ and $\sum_{i=1}^n y_i^t=0$, we have $\ex{W^{t+1}-W^{t}\mid \mathfrak{F}^t }\geq   1 + \frac{\Delta^t}{2n}$ by \eqref{eq:1diff}. Thus,
	\begin{equation}\label{eq:delt}\Ex{W^{t+1}-W^{t}\;\big| \;\mathfrak{F}^t   }\cdot \mathbf{1}_{\{\Delta^{t}  \geq n/10\}} \geq \left( 1 + \frac{1}{20}\right)\cdot \mathbf{1}_{\{\Delta^{t}  \geq n/10\}} .\end{equation} Recall that $B_-^t$ is the number of underloaded bins at time $t$. Thus by \eqref{eq:1diff} we have
	\begin{equation}\begin{aligned}\label{eq:Bbound}\Ex{W^{t+1}-W^{t}\;\big| \; \mathfrak{F}^t  } \cdot \mathbf{1}_{\{|B_{-}^{t}| \geq n/20 \}}&\geq  \left( 1 + \frac{n}{20}\cdot \frac{1}{n}\right) \cdot \mathbf{1}_{\{|B_{-}^{t}| \geq n/20 \}} \\&= \left(1+ \frac{1}{20}\right) \cdot \mathbf{1}_{\{|B_{-}^{t}| \geq n/20 \}}.\end{aligned}\end{equation}
	For any $t_0\geq 1$, $t\in[t_0,t_0+n]$, define $A$ to be the (random) set of times $t\in [t_0,t_0+n ]$ where the event $\mathcal{E}_t=\{\Delta^{t} \geq  n/10\} \cup \{|B_{-}^{t}| \geq n/20\}$ holds. By \eqref{eq:delt} and \eqref{eq:Bbound} we have $\Ex{W^{t+1}-W^{t}\;\big| \; \mathfrak{F}^t  }\mathbf{1}_{\mathcal{E}_t}\geq 1 + 1/20 $. Then, since $W^{t+1}-W^t\geq 1$ for any $t$, for any $t_0\geq 1$,  
	\begin{align*} \Ex{W^{t_0+n}-W^{t_0}} &= \Ex{ \sum_{t=t_0}^{t_0+n }\Ex{W^{t+1}-W^{t} \;\big| \; \mathfrak{F}^t}\mathbf{1}_{t\in \mathcal{E}_t}   +\Ex{W^{t+1}-W^{t} \;\big| \; \mathfrak{F}^t}\mathbf{1}_{t\notin \mathcal{E}_t}}\\ &\geq \Ex{| A|\cdot \left(1 + \frac{1}{20}\right)   + (n-| A|)  }.   \end{align*}
	Now, by \Cref{cor:enough_good_quantiles} for any $t_0\geq 1$ we have $\Pro{|A|\geq n/40}=1$ and so
	\begin{equation}\label{eq:bigt}\Ex{W^{t_0+n}-W^{t_0}}\geq \Ex{\frac{n}{40}\cdot \left(1 + \frac{1}{20}\right)   + \frac{39n}{40} }  = n\left( 1 + \frac{1}{800}\right).\end{equation}
	
	Since the bound from \eqref{eq:bigt} holds for any $t_0\geq 1$ the result follows for any $t> n$ (more details given below) however we must first consider the case $2\leq t \leq n$ separately.
	
	Observe that in the first round we place one ball. Then, until $n$ balls have been placed we place two balls if we sample an underloaded bin and this is the most we can place in any round. We have little control how these are placed but certainly for any round $r\leq \lfloor n/3\rfloor$ there are at least $n-1-2\cdot(\lfloor n/3\rfloor-1)\geq \lceil n/3\rceil  $  underloaded (empty) bins when we sample a  bin. It follows from condition \POne that for any $r\leq \lfloor n/3\rfloor$ the number of underloaded bins sampled in the first $r$ rounds (excluding the first) stochastically dominates a $\operatorname{Bin}(r-1, 1/3)$ random variable, which has median at least $\lfloor (r-1)/3\rfloor$. Thus, with probability at least $1/2$, at least $\lfloor (r-1)/3\rfloor$ of the first $r$ rounds contribute two balls. This is not greater than $0$ for $r\leq 6$, however for $2\leq r\leq 6$ at most $1+2\cdot6=13$ bins are occupied. Thus for any $2\leq r\leq 6$ two balls are assigned in $r-1\geq 1$ of the first $r$ rounds (all rounds but the first) with probability at least $1-5\cdot (13/n)\geq 1/2$ by the union bound, since we assume throughout that $n$ is sufficiently large. Thus for any $2\leq r\leq \lfloor n/3\rfloor$ we have 
	\[\Ex{W^{r}}\geq r  + \frac{1}{2}\cdot \max\left\{\lfloor (r-1)/3\rfloor, 1 \right\} \geq r\cdot \left(1+\frac{1}{12}\right).   \] 
	If we assume (pessimistically) that only one ball is allocated at any round $t> \lfloor n/3\rfloor$, then for any $t\leq n$  we have  
	\[\Ex{W^{t}} \geq \Ex{W^{\min\{t, \lfloor n/3\rfloor\}}} + \min\{0, t-\lfloor n/3\rfloor\}  \geq t\cdot \left(1+\frac{1}{50}\right).  \]
	Hence for any $2\leq t\leq  n$ we have $\ex{\mu^t }\geq \Ex{W^{t}}/t > 1+1/50$. Also for any $t> n$ we have \begin{align*}\ex{\mu^t }&\geq \frac{1}{t}\cdot \left( \sum_{i= 1}^{\lfloor (t-1)/n\rfloor} \Ex{W^{i\cdot n+1}-W^{(i-1)\cdot n+1}} + \left(t -\left\lfloor \frac{t-1}{n}\right\rfloor\cdot n   \right)\right) \\ &\geq \frac{1}{t}\cdot \left(t + \frac{1}{800}\left\lfloor \frac{t-1}{n}\right\rfloor\cdot n \right)\geq 1 + \frac{1}{1600},\end{align*} by \eqref{eq:bigt}. Thus taking $c= 1 + \frac{1}{1600}$ gives the lower bound on $\ex{\mu^t }$.
	
	For the upper bound, by \Cref{lem:supermartin}, there exists a constant $\tilde{c}>0$ such that for any $m\geq t_0\geq 0$,
	\[\Ex{W^m-W^{t_0}\; \Big|\; \mathfrak{F}^{t_0}}\leq (\tilde{c}+2) \cdot (m-t_0).\] 
	We now choose $m=t$ and $t_0=0$, and using $\ex{W^0}=0$, we conclude
	\[
	\Ex{W^t} \leq (\tilde{c}+2)\cdot t.
	\]
	Using this, and since for any process satisfying \POne and \WOne, we conclude that
	\[\ex{\mu^t}  =\mathbf{E}\left[\frac{W^t}{t}\right] = \frac{\Ex{W^t}}{t}\leq \frac{(\tilde{c}+2)\cdot t}{t} :=C,\]for the constant $C=\tilde{c}+2$.      
\end{proof}

\section{Experimental Results} \label{sec:experiemental_results}

In this section, we present some empirical results for the \Packing, \TightPacking  and \Memory processes (\Cref{fig:gap_vs_bins_alt} and \Cref{tab:gap_distribution}) and compare their load with that of a $(1 + \beta)$-process with $\beta = 0.5$, a $\Quantile(1/2)$ process, and the \TwoChoice process. 

\colorlet{GA}{black!40!white}
\colorlet{GB}{black!70!white}
\colorlet{GC}{black}
\newcommand{\CA}[1]{\textcolor{GA}{#1}} %
\newcommand{\CB}[1]{\textcolor{GB}{#1}} %
\newcommand{\CC}[1]{\textcolor{GC}{#1}} %
\newcommand{\CI}[2]{\FPeval{\result}{min(30 + 3 * #1, 100)} \colorlet{tmpC}{black!\result!white} {\textcolor{tmpC}{#2}}}

\begin{table}[H]
	\centering	\resizebox{\textwidth}{!}{
		\begin{tabular}{|c|c|c|c|c|c|c|}
			\hline
			$n$ & $(1+\beta)$ for $\beta = 1/2$ & \Packing& \TightPacking & $\Quantile(1/2)$ &  \Memory & \TwoChoice \\ \hline
			$10^3$ &
			\makecell{
				\CI{5}{\textbf{12} : \ 5\% } \\
				\CI{15}{\textbf{13} : 15\% } \\
				\CI{31}{\textbf{14} : 31\% } \\
				\CI{21}{\textbf{15} : 21\% } \\
				\CI{15}{\textbf{16} : 15\% } \\
				\CI{5}{\textbf{17} : \ 5\% } \\
				\CI{4}{\textbf{18} : \ 4\% } \\
				\CI{2}{\textbf{19} : \ 2\% } \\
				\CI{1}{\textbf{20} : \ 1\% } \\
				\CI{1}{\textbf{21} : \ 1\% } } &
			\makecell{
				\CI{3}{\textbf{\ 6} : \ 3\%} \\
				\CI{14}{\textbf{\ 7} : 14\%} \\
				\CI{30}{\textbf{\ 8} : 30\%} \\
				\CI{23}{\textbf{\ 9} : 23\%} \\
				\CI{15}{\textbf{10} : 15\%} \\
				\CI{8}{\textbf{11} : \ 8\%} \\
				\CI{4}{\textbf{12} : \ 4\%} \\
				\CI{1}{\textbf{13} : \ 1\%} \\
				\CI{1}{\textbf{14} : \ 1\%} \\
				\CI{1}{\textbf{15} : \ 1\%} } &
			\makecell{
				\CI{23}{\textbf{5} : 23\%} \\
				\CI{50}{\textbf{6} : 50\%} \\
				\CI{15}{\textbf{7} : 15\%} \\
				\CI{10}{\textbf{8} : 10\%} \\
				\CI{1}{\textbf{9} : \ 1\%} \\
				\CI{1}{\textbf{10} : \ 1\%} } & 
			\makecell{
				\CI{1}{\textbf{\ 3} : \ 1\% } \\ 
				\CI{11}{\textbf{\ 4} : 11\% } \\ 
				\CI{46}{\textbf{\ 5} : 46\% } \\
				\CI{33}{\textbf{\ 6} : 33\% } \\
				\CI{6}{\textbf{\ 7} : \ 6\% } \\
				\CI{2}{\textbf{\ 8} : \ 2\% } \\
				\CI{1}{\textbf{10} : \ 1\% } } &
			\makecell{
				\CI{67}{\textbf{2} : 67\%} \\
				\CI{33}{\textbf{3} : 33\%} } &
			\makecell{
				\CI{93}{\textbf{2} : 93\% } \\
				\CI{7}{\textbf{3} : \ 7\% }} \\ \hline
			$10^4$ &
			\makecell{
				\CI{3}{\textbf{16} : \ 3\% } \\
				\CI{21}{\textbf{17} : 21\% } \\
				\CI{19}{\textbf{18} : 19\% } \\
				\CI{10}{\textbf{19} : 10\% } \\
				\CI{23}{\textbf{20} : 23\% } \\
				\CI{11}{\textbf{21} : 11\% } \\
				\CI{10}{\textbf{22} : 10\% } \\
				\CI{2}{\textbf{23} : \ 2\% } \\
				\CI{1}{\textbf{24} : \ 1\% }} &
			\makecell{
				\CI{2}{\textbf{\ 9} : \ 2\%} \\
				\CI{17}{\textbf{10} : 17\%} \\
				\CI{28}{\textbf{11} : 28\%} \\
				\CI{14}{\textbf{12} : 14\%} \\
				\CI{22}{\textbf{13} : 22\%} \\
				\CI{11}{\textbf{14} : 11\%} \\
				\CI{3}{\textbf{15} : \ 3\%} \\
				\CI{2}{\textbf{16} : \ 2\%} \\
				\CI{1}{\textbf{17} : \ 1\%} } &
			\makecell{
				\CI{3}{\textbf{6} : \ 3\%} \\
				\CI{24}{\textbf{7} : 24\%} \\
				\CI{45}{\textbf{8} : 45\%} \\
				\CI{23}{\textbf{9} : 23\%} \\
				\CI{5}{\textbf{10} : \ 5\%} } & 
			\makecell{
				\CI{14}{\textbf{\ 6} : 14\% } \\
				\CI{42}{\textbf{\ 7} : 42\% } \\
				\CI{25}{\textbf{\ 8} : 25\% } \\
				\CI{15}{\textbf{\ 9} : 15\% } \\
				\CI{2}{\textbf{10} : \ 2\% } \\
				\CI{1}{\textbf{11} : \ 1\% } \\
				\CI{1}{\textbf{12} : \ 1\% } } &
			\makecell{
				\CI{5}{\textbf{2} : \ 5\%} \\
				\CI{95}{\textbf{3} : 95\%} } &
			\makecell{
				\CI{46}{\textbf{2} : 46\% } \\
				\CI{54}{\textbf{3} : 54\% }} \\ \hline
			$10^5$ & 
			\makecell{
				\CI{2}{\textbf{20} : \ 2\% } \\ 
				\CI{7}{\textbf{21} : \ 7\% } \\
				\CI{9}{\textbf{22} : \ 9\% } \\
				\CI{26}{\textbf{23} : 26\% } \\
				\CI{27}{\textbf{24} : 27\% } \\
				\CI{14}{\textbf{25} : 14\% } \\
				\CI{6}{\textbf{26} : \ 6\% } \\
				\CI{3}{\textbf{27} : \ 3\% } \\
				\CI{4}{\textbf{28} : \ 4\% } \\
				\CI{1}{\textbf{29} : \ 1\% } \\
				\CI{1}{\textbf{34} : \ 1\% }} &
			\makecell{
				\CI{2}{\textbf{12} : \ 2\%} \\
				\CI{16}{\textbf{13} : 16\%} \\
				\CI{20}{\textbf{14} : 20\%} \\
				\CI{28}{\textbf{15} : 28\%} \\
				\CI{23}{\textbf{16} : 23\%} \\
				\CI{5}{\textbf{17} : \ 5\%} \\
				\CI{3}{\textbf{18} : \ 3\%} \\
				\CI{1}{\textbf{19} : \ 1\%} \\
				\CI{2}{\textbf{20} : \ 2\%} } &
			\makecell{
				\CI{4}{\textbf{8} : \ 4\%} \\
				\CI{33}{\textbf{9} : 33\%} \\
				\CI{40}{\textbf{10} : 40\%} \\
				\CI{17}{\textbf{11} : 17\%} \\
				\CI{5}{\textbf{12} : \ 5\%} \\
				\CI{1}{\textbf{13} : \ 1\%} } & 
			\makecell{
				\CI{28}{\textbf{\ 8} : 28\% } \\
				\CI{42}{\textbf{\ 9} : 42\% } \\
				\CI{18}{\textbf{10} : 18\% } \\
				\CI{7}{\textbf{11} : \ 7\% } \\
				\CI{3}{\textbf{12} : \ 3\% } \\
				\CI{1}{\textbf{14} : \ 1\% } \\
				\CI{1}{\textbf{15} : \ 1\% }} &
			\makecell{
				\CI{100}{\textbf{3} : 100\%}
			} &
			\makecell{
				\CI{100}{\textbf{3} : 100\% }} \\ \hline
	\end{tabular}}
	\caption{Summary of observed gaps  for $n \in \{ 10^3, 10^4, 10^5\}$ bins and $m=1000 \cdot n$ number of balls, for $100$ repetitions. The observed gaps are in bold and next to that is the $\%$ of runs where this was observed.}
	\label{tab:gap_distribution}
\end{table}

\begin{figure}[ht]
	\centering
	\includegraphics[height=4.4cm]{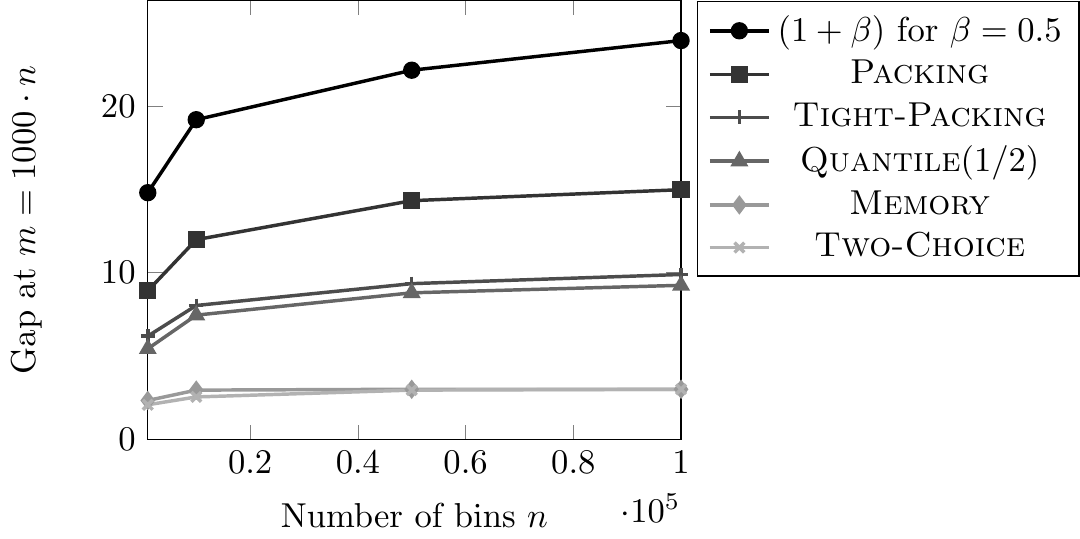}
	\caption{Average Gap vs.~$n \in \{ 10^3, 10^4, 5 \cdot 10^4, 10^5\}$ for the experimental setup of \Cref{tab:gap_distribution}.}
	\label{fig:gap_vs_bins_alt}
\end{figure}

\begin{figure}[H]
	\centering
	\includegraphics[height=4.4cm]{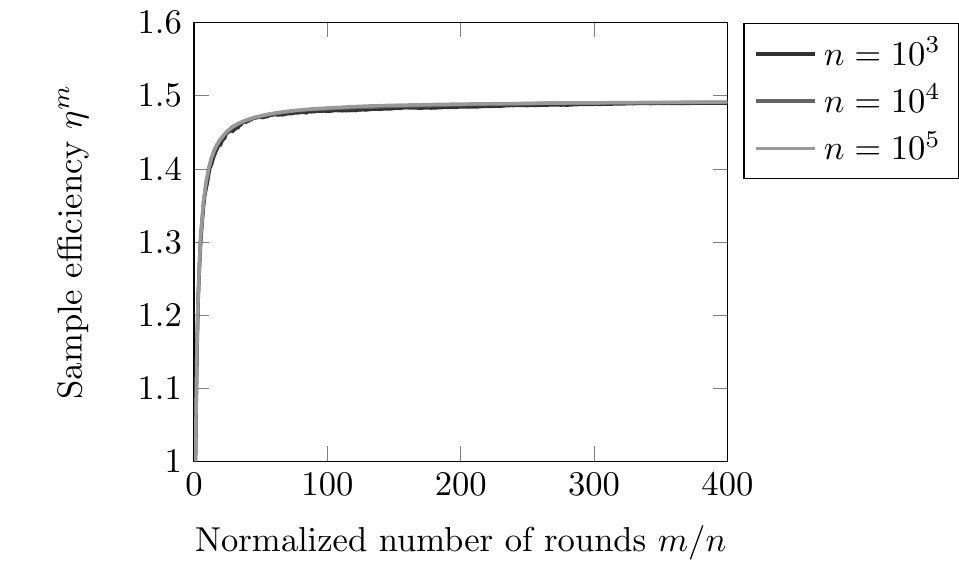}
	\caption{Sample efficiency for the \Packing process versus the number  of rounds over the number of bins $n$ for $n \in \{ 10^3, 10^4, 10^5\}$. The sample efficiency seems to converge to $1.5$.}
	\label{fig:packing_sample_efficiency}
\end{figure}

\section{Conclusions}\label{sec:conclusions}
In this work, we introduced a new class of allocation processes we call \Filling processes. Roughly speaking, these processes have a probability vector majorized by \OneChoice, they allocate one ball when they sample an overloaded bin, but when they sample an underloaded bin they can allocate the number of missing balls to some arbitrary underloaded bins. We proved that any \Filling process achieves an $\Oh(\log n)$ gap at an arbitrary round \Whp\ (\Cref{thm:filling_key}) and, for some constant $c > 0$, allocates at least $1+c$ balls in each round in expectation (\Cref{thm:sample_efficient}). 

Our prototype \Filling process is \Packing, which selects at each round a random bin, and if the bin is overloaded, allocates a single ball; otherwise it ``fills'' the underloaded bin with balls up until it becomes overloaded. For the \Packing process we proved that the general upper bound of $\mathcal{O}(\log n)$ is tight for any sufficiently large $m$ (\Cref{lem:packing_logn_lb}). A consequence of \Cref{thm:sample_efficient} is that, in contrast to other processes with a gap that does not depend on the number of balls such as \TwoChoice, $(1+\beta)$-process and $\Quantile(1/2)$, \Packing is more sample-efficient than \OneChoice. Additionally we showed that, unlike \TwoChoice, \Packing can also handle arbitrarily biased distributions.  

We also prove that our results for \Filling processes can be extended to the \Memory process by Mitzenmacher, Prabhakar and Shah~\cite{MPS02}. Using this extension we prove the first $\mathcal{O}(\log n )$ gap bound for the \Memory process with a polynomial number of balls.

There are several possible extensions to this work. One is to explore \emph{stronger} versions of the conditions on the probability vector which imply $o(\log n)$ gap bounds. For example, it might be interesting to explore a version of \Packing where the bin is sampled using \TwoChoice.

At the opposite end, one might investigate probability vectors with \emph{weaker} guarantees. In \Cref{sec:packing_with_negative_bias}, we showed that for the \Packing process with an arbitrarily $(a, b)$-biased sampling vector (which may even majorize \OneChoice) the gap is \Whp~at most $a \cdot \poly(n)$, i.e., still independent of $m$. This demonstrates the ``power of filling'' in balanced allocations.

\section*{Acknowledgments} 
We thank David Croydon and Martin Krejca for some helpful discussions.

\setlength{\bibsep}{3pt plus 3ex}

\clearpage

\appendix

\section{Auxiliary Inequalities}

We provide an elementary proof of the following lemma for completeness; our proof is similar to that of~\cite[Lemma A.1]{FGS12}. This inequality also appears in \cite[Ch.~XII]{MPFbook} where the authors state that it is a consequence of Abel's transformation (summation by parts).  

\begin{lem}\label{lem:quasilem}Let the real sequences  $(a_k)_{k=1}^n $ and $(b_k)_{k=1}^n $ be non-negative, and $(c_k)_{k=1}^n$ be non-negative and non-increasing. If $\sum_{k=1}^i a_k \leq \sum_{k=1}^i b_k$ holds for all $1\leq i\leq n$ then, \begin{equation} \label{eq:toprove}\sum_{k=1}^n a_k\cdot c_k \leq \sum_{k=1}^n b_k\cdot c_k.\end{equation}
\end{lem}
\begin{proof}
	We shall prove \eqref{eq:toprove} holds by induction on $n\geq 1$. The base case $n=1$ follows immediately from the fact that $a_1\leq b_1$ and $c_1\geq 0$. Thus we assume $\sum_{k=1}^{n-1} a_k\cdot c_k \leq \sum_{k=1}^{n-1} b_k\cdot c_k$ holds for all sequences $(a_k)_{k=1}^{n-1}, (b_k)_{k=1}^{n-1}$ and $ (c_k)_{k=1}^{n-1}$ satisfying the conditions of the lemma.
	
	For the inductive step, suppose we are given sequences $(a_k)_{k=1}^{n}$,  $(b_k)_{k=1}^{n}$ and $(c_k)_{k=1}^{n}$ satisfying the conditions of the lemma. If $c_2=0$ then, since $(c_k)_{k=1}^n$ is non-increasing and non-negative, $c_k=0$ for all $k\geq 2$. Thus as $a_1\leq b_1$ and $c_1\geq 0$ by the precondition of the lemma, we conclude 
	\[ \sum_{k=1}^n a_k\cdot c_k = a_1 \cdot c_1 \leq b_1\cdot c_1 =\sum_{k=1}^n b_k\cdot c_k. \] We now treat the case $c_2>0 $. Define the non-negative sequences $(a_k')_{k=1}^{n-1}$ and $(b_k')_{k=1}^{n-1}$ 
	as follows:
	\begin{itemize}
		\item $a'_{1} = \frac{c_1}{c_{2}} \cdot a_{1} + a_2$ and $a'_{k} = a_{k+1}$ for $2\leq k \leq n-1$ ,
		\item $b'_{1} = \frac{c_1}{c_{2}} \cdot b_{1} + b_2$ and $b'_{k} = b_{k+1}$ for $2\leq k \leq n-1$,
	\end{itemize} Then as the inequalities $c_1\geq c_2 $, $a_1 \leq b_1$ and $\sum_{i=1}^n a_k \leq \sum_{i=1}^n b_k$ hold by assumption, we have  
	\[ \sum_{k=1}^{n-1}a_k' = \left(\frac{c_{1}}{c_{2}} - 1  \right)a_1 + \sum_{k=1}^{n}a_k  \leq \left(\frac{c_{1}}{c_{2}} - 1  \right)b_1 +  \sum_{k=1}^{n}b_k   = \sum_{k=1}^{n-1}b_k'. \] Thus if we also let $(c_k')_{k=1}^{n-1} = (c_{k+1})_{k=1}^{n-1} $, which is positive and non-increasing, then 
	\[\sum_{k=1}^{n-1} a_k'\cdot c_k'  \leq \sum_{k=1}^{n-1} b_k'\cdot c_k',\]by the inductive hypothesis. However \[\sum_{k=1}^{n-1} a_k'\cdot c_k' = \left(\frac{c_1}{c_{2}} \cdot a_{1} + a_2 \right)c_2 +  \sum_{k=2}^{n-1} a_{k+1}\cdot c_{k+1} = \sum_{k=1}^{n} a_k\cdot c_k, \] and likewise $\sum_{k=1}^{n-1} b_k'\cdot c_k' = \sum_{k=1}^{n} b_k\cdot c_k$. The result follows.
\end{proof}

Again, for completeness, we define Schur-convexity (see \cite{MRBook}) and state two basic results:

\begin{defi}
	A function $f : \mathbb{R}^n \to \mathbb{R}$ is Schur-convex if for any non-decreasing $x, y \in \mathbb{R}^n$, if $x$ majorizes $y$ then $f(x) \geq f(y)$. A function $f$ is Schur-concave if $-f$ is Schur-convex.
\end{defi}

\begin{lem} \label{lem:sum_of_convex_is_schur_convex}
	Let $g : \mathbb{R} \to \mathbb{R}$ be a convex (resp. concave) function. Then, the function $g(x_1, \ldots, x_n) := \sum_{i = 1}^n g(x_i)$ is Schur-convex (resp. Schur-concave). 
\end{lem}

\begin{lem}\label{clm:schur}
	For any $\alpha > 0$, for any $\beta \in \mathbb{R}$ and any $\Delta \in \mathbb{R}$, consider the function \[f(x_1,x_2,\ldots,x_k) = \sum_{j=1}^{k} \exp\left(- \alpha x_j \right),\] where $\sum_{j=1}^{k} x_j \geq  \Delta$ and $x_{j} \geq \beta $ for all $1 \leq j \leq k$. Then,
	\[
	f(x_1,x_2,\ldots,x_k) \leq (k-1) \cdot \exp\left(- \alpha \cdot \beta \right) + 1 \cdot \exp\left(- \alpha \cdot \left( \Delta - (k-1) \cdot \beta \right) \right).
	\]
\end{lem}
\begin{proof}
	Note that by \Cref{lem:sum_of_convex_is_schur_convex}, it follows that $f(x_1, \ldots, x_n)= \sum_{i=1}^n g(x_i)$ is Schur-concave, since $g(z) = e^{-\alpha z}$ is concave for $\alpha > 0$. As a consequence, the function attains its maximum if the values $(x_1,x_2,\ldots,x_k)$ are as ``spread out'' as possible, i.e., if any prefix sum of the values ordered non-increasingly is as large as possible.
\end{proof}

\begin{lem} \label{lem:geometric_arithmetic}
	Consider any sequence $(z_i)_{i \in \mathbb{N}}$ such that, for some $a > 0$ and $b > 0$, for every $i \geq 1$,
	\[
	z_i \leq z_{i-1} \cdot a + b.
	\]
	Then for every $i \in \mathbb{N}$, 
	\[
	z_i \leq z_0 \cdot a^i + b \cdot \sum_{j = 0}^{i-1} a^j.
	\]
	Further, if $a < 1$, then
	\[
	z_i 
	\leq z_0 \cdot a^i + \frac{b}{1 - a}.
	\]
\end{lem}
\begin{proof}
	We will prove the first claim by induction. For $i = 0$, $z_0 \leq z_0$. Assume the induction hypothesis holds for some $i \geq 0$, then since $a > 0$,
	\[
	z_{i+1} \leq z_{i} \cdot a + b \leq \Big(z_0 \cdot a^i + b \cdot \sum_{j = 0}^{i-1} a^j \Big) \cdot a + b = z_0 \cdot a^{i+1} +b \cdot \sum_{j = 0}^i a^j.
	\]
	Hence, the first claim follows. The second part of the claim is immediate, since for $a \in (0,1)$, $\sum_{j=0}^{\infty} a^j = \frac{1}{1-a}$.
\end{proof}

\section{Counterexample for the Exponential Potential Function}\label{sec:potentialgoesup}

In this section, we present a configuration for which the exponential potential function $\Phi$ (such as the one defined in \Cref{sec:analysis_filling}), increases by a multiplicative factor in expectation, in a round $t$ where the ``good event'' $\mathcal{G}^t$ does not hold.

\begin{clm}\label{clm:counterexample}
	For any constant $\alpha > 0$ and for sufficiently large $n$, consider the (normalized) load configuration,
	\[
	y^t = (\sqrt{n}, \underbrace{0, \ldots, 0}_{n-\sqrt{n}-1\text{ bins}}, \underbrace{-1, \ldots, -1}_{\sqrt{n}\text{ bins}}).
	\]
	Then, for the \Packing process, the potential function $\Phi^t := \sum_{i : y_i^t \geq 0} e^{\alpha y_i^t}$ will increase in expectation, i.e.,
	\[
	\ex{\Phi^{t+1} \mid \mathfrak{F}^t} \geq \Phi^t \cdot \Big(1 + 0.1 \cdot \frac{\alpha^2}{n}\Big).
	\]
\end{clm}
\begin{proof}
	Consider the contribution of bin $i=1$, with $y_1^t = \sqrt{n}$.
	\begin{align*}	&\ex{\Phi_1^{t+1} \mid \mathfrak{F}^t}\\ 
		&\quad  = e^{\alpha \sqrt{n}} \cdot \Big( 1 + \frac{1}{n} \cdot (e^{\alpha - \alpha/n} - 1) + \frac{n - \sqrt{n} - 1}{n} \cdot (e^{-\alpha/n} -1) + \frac{\sqrt{n}}{n} \cdot (e^{-2\alpha/n} -1) \Big).\end{align*} 
	Now using a Taylor estimate $e^z \geq 1 + z + 0.3 z^2$ for $z \geq -1.5$, 
	\begin{align*}
		\ex{\Phi_1^{t+1} \mid \mathfrak{F}^t}
		& \geq e^{\alpha \sqrt{n}} \cdot \Big( 1 + \frac{1}{n} \cdot \Big(\alpha - \frac{\alpha}{n} + 0.3 \cdot \Big( \alpha - \frac{\alpha}{n}\Big)^2\Big)  \\
		& \qquad + \frac{n - \sqrt{n} - 1}{n} \cdot \Big(-\frac{\alpha}{n} + 0.3 \cdot \frac{\alpha^2}{n^2}\Big) + \frac{\sqrt{n}}{n} \cdot \Big(-\frac{2\alpha}{n} + 1.2 \cdot \frac{\alpha^2}{n^2}\Big) \Big) \\
		& = e^{\alpha \sqrt{n}} \cdot \Big( 1 + \frac{\alpha + 0.3 \cdot \alpha^2}{n} -\frac{\alpha}{n} + o(n^{-1}) \Big) \\ &= e^{\alpha \sqrt{n}} \cdot \Big( 1 + 0.3 \cdot \frac{\alpha^2}{n} + o(n^{-1}) \Big) \\
		& \geq e^{\alpha \sqrt{n}} \cdot \Big( 1 + 0.2 \cdot \frac{\alpha^2}{n} \Big).
	\end{align*}
	At round $t$, the contribution of the rest of the bins is at most $n$, i.e., $\sum_{i > 1, y_i^t \geq 0} \Phi^t \leq n$. Note that since $\alpha$ is a constant for sufficiently large $n$, we have $n \cdot (1 + 0.1 \cdot \frac{\alpha^2}{n})  < 0.1 \cdot \alpha^2 \cdot e^{\alpha \sqrt{n}}$. Hence,
	\begin{align*}
		\ex{\Phi^{t+1} \mid \mathfrak{F}^t} 
		& \geq e^{\alpha \sqrt{n}} \cdot \Big( 1 + 0.2 \cdot \frac{\alpha^2}{n} \Big)
		= e^{\alpha \sqrt{n}} \cdot \Big( 1 + 0.1 \cdot \frac{\alpha^2}{n} \Big) + 0.1 \cdot \frac{\alpha^2}{n} \cdot e^{\alpha \sqrt{n}} \\ 
		& \geq \Phi_1^t \cdot \Big( 1 + 0.1 \cdot \frac{\alpha^2}{n} \Big) + n \cdot (1 + 0.1 \cdot \frac{\alpha^2}{n}) \\
		& \geq \Phi_1^t \cdot \Big( 1 + 0.1 \cdot \frac{\alpha^2}{n} \Big) + \left(\sum_{i > 1, y_i^t \geq 0} \Phi^t \right) \cdot \Big( 1 + 0.1 \cdot \frac{\alpha^2}{n} \Big) \\
		& = \Phi^t \cdot \Big( 1 + 0.1 \cdot \frac{\alpha^2}{n} \Big),
	\end{align*}as claimed.
\end{proof}

\end{document}